\documentclass{article}
\usepackage[letterpaper,head=13pt, marginparwidth=6pc,heightrounded]{geometry}
\usepackage{booktabs}
\usepackage[ruled]{algorithm2e}
\usepackage{amsthm}
\usepackage{amsmath}
\usepackage{thmtools, thm-restate}
\usepackage[bookmarksnumbered,unicode]{hyperref}
\usepackage{tikz}
\usetikzlibrary{positioning}

\usepackage{graphicx}
\usepackage{amsfonts}

\SetAlFnt{\small}
\SetAlCapFnt{\small}
\SetAlCapNameFnt{\small}
\SetAlCapHSkip{0pt}
\IncMargin{-\parindent}

\newtheorem{theorem}{Theorem}
\newtheorem{definition}{Definition}
\newtheorem{remark}{Remark}
\newtheorem{lemma}{Lemma}

\newtheorem{example}{Example}

\newtheorem{take-away}{Take Away}
\newtheorem{informal-theorem}{Informal Theorem}
\newtheorem{claim}{Claim}
\newcommand{\lw}{\text{LW}}

\newcommand{\x}{\mathbf{x}}

\usepackage{natbib}
\setcitestyle{authoryear}
\title{Revenue Guarantees in Autobidding Platforms}
\author{Ioannis Caragiannis, Anders Bo Ipsen, Stratis Skoulakis \\ {\small Department of Computer Science, Aarhus University}}
\date{February 16, 2026}

\begin{document}

\begin{titlepage}
  \maketitle
  \vspace{-1cm}
  \begin{abstract}
    Motivated by autobidding systems in online advertising, we study revenue maximization in markets with divisible goods and budget-constrained buyers with linear valuations. Our aim is to compute a single price for each good and an allocation that maximizes total revenue. We show that the First-Price Pacing Equilibrium (FPPE) guarantees at least half of the optimal revenue, even when compared to the maximal revenue of buyer-specific prices. This guarantee is particularly striking in light of our hardness result: we prove that revenue maximization under individual rationality and single-price-per-good constraints is APX-hard.

    We further extend our analysis in two directions: first, we introduce an online analogue of FPPE and show that it achieves a constant-factor revenue guarantee, specifically a $1/4$-approximation; second, we consider buyers with concave valuation functions, characterizing an FPPE-type outcome as the solution to an Eisenberg–Gale-style convex program and showing that the revenue approximation degrades gracefully with the degree of nonlinearity of the valuations.
  \end{abstract}
  \setcounter{tocdepth}{1}
  \tableofcontents

\end{titlepage}

\section{Introduction}

The bidding process in online advertising markets has become increasingly complex~\cite{survey}.
On modern platforms, advertisers typically participate in a very large number of auctions
simultaneously, each corresponding to a different keyword.
An advertiser may be interested in multiple keywords over a given time window, may
value these keywords differently, and is almost always subject to a global budget
constraint that its limits total spending over the duration of a campaign.
The advertiser’s objective is to allocate this budget across keywords and time so as
to maximize utility subject to these constraints~\cite{castiglioni22a,castiglioni2024online,LPS24a}.

Historically, advertising platforms sold keywords through independent auctions—often
first-price or second-price—run separately for each keyword and each impression.
As the number of auctions and constraints grew, however, this approach became
impractical for advertisers to manage directly.
To address this complexity, platforms introduced \emph{autobidding} systems, in which
advertisers specify high-level campaign goals and the platform bids on their behalf
in individual auctions~\cite{D21,LS23,PM23,LPS24a,G23a,DGPV24,deng2024efficiency}.
Typical high-level inputs provided by advertisers include: \emph{What is the starting and the end date of the campaign? What is the total budget? What is the maximum amount that the advertiser is willing to pay for an impression of a specific keyword?}
\smallskip

Given this information, an autobidding system serves as a proxy bidder for the advertiser.
A dominant paradigm for implementing autobidding strategies is through the use of
\emph{pacing multipliers}: the platform assigns each advertiser a multiplier in $[0,1]$ that scales all of the advertisers' values over the different keywords~\cite{CKPSSMW19,LKLSSMZ25,LK24}.
These multipliers are chosen so that the total spending of the each advertiser respects their overall budget.
\smallskip
\smallskip

\noindent
\textbf{From Autobidding to Divisible Goods.}
In their seminal work,
\cite{CKPSSMW19} reduce the problem of designing an autobidding mechanism to the problem of selling \textit{multiple divisible goods} to \textit{multiple budget-constrained buyers} via introducing the notion of \textit{First-Price Pacing Equilibrium}~(FPPE) (also considered in \cite{BCIJEM07}).

In particular, FPPE models each keyword as a \emph{divisible good}, an assumption well-justified by the high volume of impressions that allows supply to be fractionally partitioned among advertisers. Each advertiser reports: $i)$ its total budget and $ii)$ the maximum value it is willing to pay for an impression of each keyword. The FPPE then determines a fractional allocation of impressions and a market-clearing price for each keyword. Crucially, the FPPE ensures the following utility-maximization property for each advertiser:
\smallskip
\smallskip

\noindent \textit{The allocation assigned to each advertiser is the utility-maximizing allocation with respect to the keyword prices subject to the budget declared by the advertiser}.
\smallskip
\smallskip

\noindent \cite{CKPSSMW19} perform empirical evaluations using real-world data and conclude that advertisers have a very small incentive to misreport their actual valuations and budget. More recently, \cite{YKLS24} establish that the Price of Anarchy (PoA) with respect to the Liquid Social Welfare is at most $4$ for any mixed Nash Equilibrium. These results motivate the following natural question:

\begin{center}
  \emph{How does the revenue of FPPE compare to the optimal revenue achievable under individual rationality and budget feasibility constraints?}
\end{center}

\noindent
\subsection{Our Contributions and Results.}
Taking a step back from the FPPE, consider the setting of a profit-seeking seller offering multiple divisible goods to budget-constrained \textit{non-strategic} buyers. To maximize revenue, the seller would calculate the allocation and payment vectors that maximize the attained revenue subject to individual rationality and budget spending constraints. This optimization can be efficiently computed by solving an suitable linear program.
\begin{example}[Optimal revenue with variable unit prices]\label{ex:1}
  Consider a single divisible good and two buyers.
  Buyer~1 has valuation $v_{11} = \$10$ and budget $B_1 = \$6$, while buyer~2 has valuation
  $v_{21} = \$4$ and budget $B_2 = \$4$.
  The revenue-maximizing allocation (respecting budgets and individual rationality)
  assigns $60\%$ of the good to buyer~1 for a payment of \(\$6\), and $40\%$ of the good to
  buyer~2 and for a payment of \(\$1.6\).
  The total revenue is \(\$7.6\).
\end{example}

\noindent This example highlights a significant drawback of the revenue-maximizing mechanism:
\smallskip

\quad \quad \textit{It assigns different unit prices to different buyers for the exact same good!}
\smallskip

\noindent Notice in Example~\ref{ex:1}, the unit price for buyer $1$ is $\$10$, whereas the unit price for buyer~2 is only  $\$4$. Different unit prices across buyers based on their \textit{declared} valuations and budgets is highly undesirable and, in many settings, may even lead to legal or regulatory issues.

Consequently, motivated by transparency, fairness, and implementability, it is reasonable to assume that in most realistic settings the seller is restricted to \emph{fixed unit prices}, meaning that each good must be sold at the same per-unit price to all buyers.

\begin{example}[Optimal revenue with fixed unit prices]\label{ex:2}
  In the previous example, maximizing revenue under a fixed unit price results in setting
  the unit price at \(\$6\) and allocating the entire good to buyer~1. The maximum revenue achievable with fixed unit price is therefore $\$6 < \$ 7.6$.
\end{example}

\noindent Although the problem of maximizing revenue from divisible items among budget-constrained buyers using \textit{fixed unit prices} is a fundamental problem, particularly given its applications to autobidding systems, it has not, to the best of our knowledge, been considered  previously.

Our first contribution is to establish that this problem is APX-hard - meaning that a constant-factor approximation is the best we can hope for. Our proof leverages a reduction from the 3D-2-Matching problem. This inapproximability result contrasts sharply with revenue maximization under variable unit prices, which can be solved efficiently via linear programming.

In light of this negative result, FPPE emerges as a natural heuristic for revenue
maximization under fixed unit prices. This is because FPPE always yields a fixed unit-price per good while it can be efficiently computed via solving an suitable convex program. Our main contribution is establishing that FPPE always achieves at least $50\%$ of the optimal revenue under variable unit prices!

\begin{take-away}
  FPPE achieves at least $50\%$ of the optimal revenue of the variable-price setting while restricted to fixed unit prices. Given that optimizing revenue under fixed prices is APX-hard, FPPE serves as the best polynomial-time approximation algorithm that we currently know.
\end{take-away}

Surprisingly enough, $50\%$ of the revenue is an unavoidable price that the seller has to pay for using \textit{fixed} instead of \textit{variable} unit prices. We demonstrate this tightness by constructing instances where the optimal revenue under fixed unit prices is exactly $50\%$ of the optimal revenue under variable unit prices.

We additionally extend our revenue guarantees of FPPE to concave valuation functions. While the notion of FPPE is tailored to linear valuation functions, we show that its natural generalizations to concave valuation functions admit revenue guarantees that depend solely on the form of the valuation functions.
\smallskip
\smallskip

\noindent \textbf{Extension to the Online Case.} In many digital advertising platforms, advertisers specify budgets, and \textit{start} and \textit{end} dates for their campaigns in addition to their valuations. The overall declared budget refers to their total spending over the whole campaign. See for example Fig.~$1$ in \cite{CKPSSMW19}\footnote{beginning of page~3, \textcolor{blue}{https://arxiv.org/pdf/1811.07166}} exhibiting the interface for creating a Facebook ad\footnote{\textcolor{blue}{\url{https://www.facebook.com/business/ads}}}. Motivated by this, we introduce and study an \textit{online version} of the problem where buyers arrive online and are active during specific time-intervals.

We model time as a sequence of $T$ discrete days. On each day $t$, the seller must fractionally allocate the goods among the currently active buyers. Each buyer $i$ specifies an interval $[s_i,t_i]$ at which it is active. Importantly, on day $t$, the seller is aware only of buyers with $s_i \leq t$; buyers whose intervals begin in the future remain unknown.

First-Price Pacing Equilibrium (FPPE), originally defined for static settings, extends naturally to the online setting: on each day $t$, the platform computes an FPPE using the buyers with active campaigns and their remaining budgets, and allocates that day's impressions accordingly. The second main contribution of our work is establishing that the above online algorithm is $1/4$-competitive with respect to the optimal revenue achievable with variable unit prices.

\begin{take-away}
  In case of online arrivals, computing an FPPE at each time-step based on the active buyers (with respect to their remaining budgets) leads to fixed unit-prices at each time-step and constitutes a $1/4$-competitive algorithm.
\end{take-away}

\subsection{Related Work}
Apart from the already mentioned works, our work also relates to \cite{BCIJEM07} as the equilibrium state of budget adjusted dynamics. \cite{DP14} consider the allocation rule of FPPE in the single-item case and derive an $2$-approximation truthful mechanism on the Liquid Welfare. \cite{chen21} show that the problem of finding a pacing equilibrium for a second-price auction is PPAD complete. \cite{CK22} examine pacing equilibria and prove that computing the revenue-maximizing or social welfare maximizing equilibrium is NP-hard.

A growing body of work has focused on the efficiency properties of various autobidding settings, typically through price-of-anarchy guarantees. \cite{LPS24a} provide an optimal bidding strategy in the case of multi-slot auction and show of PoA of $2$ with respect to values in case the autobidders come to an equilibrium. \cite{AMP23} show that neither FPA and SPA are incentive compatible with respect to autobidding. \cite{BDJMS21} shows that reserve prices can improve both revenue and social welfare in case of value maximizing buyers with return on investment constraints. \cite{BYJVS21} characterize optimal mechanisms for selling a good to one of multiple buyers with return on spending constraints. \cite{D24} study the price of anarchy for generalized second-price auctions in case of value maximizers. \cite{deng2024efficiency} study the price of anarchy when buyers are either utility of value maximizers. \cite{CKKLST25} study PoA bounds on the Liquid Social Welfare of budget-constrained bidders in simultaneous first-price auctions.
See also \cite{survey} for a survey on the autobidding literature.

Another line explores online learning and budget management strategies in the context of repeated auctions as well as their convergence properties. \cite{BKMV17} explore various budget management mechanisms and examine their effect on the sellers revenue and the buyers' utility while \cite{balseiro20a} use dual mirror descent to design budget management strategies. \cite{castiglioni2024online,castiglioni22a} develop algorithms online learning algorithms under budget spending and ROI constraints. \cite{BY19} study pacing strategies and show that in large markets they form an approximate Nash Equilibrium. \cite{FT25} provide liquid welfare guarantees in case of no-regret algorithms in repeated auctions with budgets. \cite{LPSSZ24} examine the emergence of complex dynamics in autobidding systems while \cite{LS23} provide conditions under which the iterated best-response converges in certain autobidding systems converges to equilibrium.

Our work also relates to the works exploring selling multiple goods. \cite{AC04} consider clock auctions to sell divisible goods. \cite{PK16} study show that the pay-as-bid auction in case of single divisible good admits a unique Bayes Pure Nash Equilibrium.

\section{Preliminaries and Paper Organization}\label{s:prelim}
Consider a seller willing to sell $m$ perfectly divisible goods to $n$ budget-constrained buyers. Each buyer $i\in [n]$ has an overall budget $B_i \geq 0$ and valuations $v_{ij} \geq 0$ for each good $j\in [m]$. Since the goods are divisible, the seller can assign to buyer $i \in [n]$ a fraction $x_{ij} \in [0,1]$ of good $j \in [m]$. In this case buyer $i \in [n]$ receives a value of $v_{ij} \cdot x_{ij}$. All buyers have \textit{individual rationality} meaning that $b_{ij} \leq v_{ij} \cdot x_{ij}$ where $b_{ij} > 0$ is the amount of money buyer $i \in [n]$ pays for good $i \in [n]$.

A seller that wants to maximize its revenue needs to choose the buyers' allocations/payments vectors $(\mathbf{x},\mathbf{b})$ as the solution of the linear program~(LP) below:
\begin{align*}
  \text{(RMVUP)} \qquad
  \max \quad
  & \sum_{i=1}^n \sum_{j=1}^m b_{ij}
  &                                  &                                                                             \\[6pt]
  \text{s.t.} \quad
  & \sum_{i=1}^n x_{ij} \le 1
  &                                  & \forall j \in [m]                                                           \\[6pt]
  & \sum_{j=1}^m b_{ij} \le B_i
  &                                  & \forall i \in [n] \quad \text{(budget constraint)}                          \\[6pt]
  & b_{ij} \le v_{ij} \cdot x_{ij}
  &                                  & \forall i \in [n],\ \forall j \in [m] \quad \text{(individual rationality)} \\[6pt]
  & x_{ij} \ge 0,\quad b_{ij} \ge 0
  &                                  & \forall i \in [n],\ \forall j \in [m]
\end{align*}

The objective value $\sum_{i\in [n]}\sum_{j \in [m]} b_{ij}$ denotes the seller's revenue. At the same time, the LP above guarantees that both individual rationality and budget feasibility for every buyer. We refer to the latter LP as \textit{Revenue Maximization with Variable Unit Prices} (RMVUP) since it can lead to \textit{different unit prices} across buyers.

The unit-price of an buyer $i \in [n]$ for good $j \in [m]$ is defined as,
\[ p_{ij} := \frac{b_{ij}}{x_{ij}},\]
namely the amount of money $b_{ij}$ buyer $i$ paid for the obtaining $x_{ij} \in [0,1]$ percent of good $j$. The solution $(\mathbf{x},\mathbf{b})$ of RMVUP leads to different unit prices across buyers, $p_{ij} \neq p_{i'j}$ (see Example~\ref{ex:1}). Charging different buyers with different unit-prices is a highly undesirable situation that raises fairness, transparency and even legal issues.

\subsection{Revenue Maximization via Fixed Unit Prices}\label{ss:RMFUP}
Due the reasons outlined above, in most real-world situations a seller will try to maximize its revenue while keeping a \textit{fixed price} $p_j > 0$ for each good $j \in [m]$. Therefore the seller needs to consider the more restrictive optimization problem in which each good $j \in [m]$ is sold at a single price $p_j \geq 0$, independent of the buyer. The latter can be done via solving the following optimization problem to which we refer as \textit{Revenue Maximization with Fixed Unit Prices} (RMFUP).

\begin{align*}
  \text{(RMFUP)} \qquad
  \max \quad
  & \sum_{i=1}^n \sum_{j=1}^m b_{ij}
  &                                                 &                                                                             \\[6pt]
  \text{s.t.} \quad
  & \sum_{i=1}^n x_{ij} \le 1
  &                                                 & \forall j \in [m]                                                           \\[6pt]
  & \sum_{j=1}^m b_{ij} \le B_i
  &                                                 & \forall i \in [n] \quad \text{(budget constraint)}                          \\[6pt]
  & b_{ij} \le v_{ij} \cdot x_{ij}
  &                                                 & \forall i \in [n],\ \forall j \in [m] \quad \text{(individual rationality)} \\[6pt]
  & \textcolor{blue}{ b_{ij} = p_j \cdot x_{ij}}
  &                                                 & \forall i \in [n],\ \forall j \in [m]                                       \\[6pt]
  & x_{ij} \ge 0,\quad b_{ij} \ge 0,\quad p_j \ge 0
  &                                                 & \forall i \in [n],\ \forall j \in [m].
\end{align*}

In RMFUP the seller needs to decide on the price $p_j$ of each good $j \in [n]$ while the constraint $b_{ij} = p_j \cdot x_{ij}$ guarantees that each buyer pays proportionally to the received quantity $x_{ij}$. Since the latter constraint is bilinear, the resulting RMFUP is a non-convex program.

\subsection{First-Price Pacing Equilibrium}
The notion of \textit{First-Price Pacing Equilibrium} (FPPE) was proposed in \cite{CKPSSMW19} as a mechanism to allocate and price divisible goods (such as impressions of keywords). In particular, an FPPE computes allocation/price vectors $(\mathbf{x},\mathbf{p})$ as well as set of \textit{pacing multipliers} $\mathbf{\alpha} := (\alpha_1,\ldots,\alpha_n) \in [0,1]$ such that the following properties hold,

\begin{enumerate}
        \smallskip
        \smallskip

  \item $\sum_{i=1}^n x_{ij} \leq 1$.
        \smallskip
        \smallskip

  \item $\sum_{j =1}^m p_j \cdot x_{ij} \leq B_i$.

        \smallskip
        \smallskip

  \item If $x_{ij} > 0$ then $p_j = a_i \cdot v_{ij}$

        \smallskip
        \smallskip

  \item $p_j := \max_{i \in [n]} a_i \cdot v_{ij}$ for each good $j \in [m]$.
        \smallskip
        \smallskip

  \item If $p_j > 0$ then $\sum_{i = 1}^n x_{ij} = 1$.

        \smallskip
        \smallskip

  \item If $\sum_{j =1}^m p_j \cdot x_{ij} < B_i$ then $a_i = 1$~~~~(no unnecessary pacing).

        \smallskip
        \smallskip
\end{enumerate}

We first remark that, due to Properties~$(1)$--$(3)$, an FPPE\footnote{We define a payments vector \( \mathbb{b} \) as \( b_i := \sum_{j=1}^m x_{ij} \cdot p_j \) and include it and the pacing multipliers in the definition of an FPPE so they are more easily referred to.} \((\mathbf{x},\mathbf{b},\mathbf{p},\mathbf{\alpha})\) satisfies the constraints of RMFUP. The pacing multipliers \(\alpha_i \in [0,1]\) uniformly scale all values \(v_{ij}\) of buyer \(i\) and can be interpreted as the bids that buyers submit in a system of simultaneous first-price auctions, where buyer \(i\)’s bid for good \(j \in [m]\) is \(\alpha_i \cdot v_{ij}\). Property~$4$ ensures that the price of each good equals the highest received bid, $p_j = \max_{i \in [n]} \alpha_i \cdot v_{ij}$. Property~$3$ guarantees that each good is allocated (not necessarily equally) among the buyers submitting this highest bid, and Property~$5$ ensures that any good with a positive price is fully allocated. Under this interpretation, each buyer \(i \in [n]\) selects its multiplier \(\alpha_i\) to maximize its total value subject to its budget constraint \(B_i\), and Property~$6$ guarantees that buyers who do not exhaust their budgets do not pace their values, i.e., \(\alpha_i = 1\). This defines a formal one-shot game with strategy space given by the multipliers \(\{\alpha_i\}_{i \in [n]}\), and an FPPE can be viewed as a Nash equilibrium of this game, which is precisely the perspective originally used to define FPPE~\cite{BCIJEM07}.

Apart from its intuitive connection with simultaneous first-price auctions, FPPE should be primarily viewed as a \textit{mechanism} that allocates divisible goods via fixed unit prices. \cite{CKPSSMW19} established the following crucial properties:

\begin{itemize}
        \smallskip
        \smallskip

  \item First-Price Pacing Equilibrium is always unique and can be efficiently computed via solving an appropriate convex program.
        \smallskip
        \smallskip

  \item First-Price Pacing Equilibrium ensures \textit{utility-maximizing allocation} for each buyer $i \in [n]$. In particular, given the computed prices $\mathbf{p}:= (p_1,\ldots,p_m)$
        \[\mathbf{x}_i := \mathrm{argmax}_{\mathbf{x}_i} \left \{ \sum_{j \in [m]} (v_{ij} - p_j)\cdot x_{ij} ~:~\sum_{j\in [m]}p_j \cdot x_{ij} \leq B_i\right \} \]
\end{itemize}
\noindent The utility maximization property above incentivizes buyers to report their true valuations and budgets since the utility maximization guarantee is reported with respect to the declared data.

\subsection{Results and Paper Organization.} In Section~\ref{s:inappr}, we establish that RMFUP is an APX-hard problem. This implies that a seller seeking to maximize revenue under fixed unit prices cannot compute revenue-maximizing prices in polynomial time, and that even a constant-factor approximation is the best achievable guarantee.

In light of this negative result, FPPE acquires an additional important role—beyond its utility-maximization properties—as an approximation mechanism for RMFUP. In Section~\ref{s:revenue_linear}, we show that FPPE achieves at least \(50\%\) of the revenue of RMVUP (which upper-bounds RMFUP). Moreover, we prove that this \(50\%\) factor is a tight worst-case approximation bound for any mechanism based on fixed unit prices with respect to RMVUP.

Motivated by the fact that advertising campaigns typically last for multiple days and overlap in time, in Section~\ref{s:online} we introduce an online version of the problem in which buyers arrive over time and specify start and end dates for their campaigns. We show that FPPE induces a \(\frac{1}{4}\)-competitive online algorithm when compared against the offline RMVUP benchmark.

Finally, in Section~\ref{s:concave}, we extend the notion of FPPE to the more general setting of concave valuation functions and establish analogous revenue guarantees.

\section{Inapproximability of Revenue Maximization with Fixed Unit Prices}\label{s:inappr}
In the section we establish that RMFUP is an APX-hard problem. In particular we show that unless P=NP, it cannot be approximated with ratio better than $854/855$. Our reduction uses the 3D-2-Matching that is known to be APX-hard~\cite{CC06}.
\begin{definition}
  In $3\mathrm{D}$-$\mathrm{Matching}$ (3DM) we are given three sets of elements $E_1,E_2,E_3$ and a set of triplets $S \subseteq E_1 \times E_2 \times E_3$ where each triplet $(a,b,c) \in S$ is composed by three elements $a\in E_1,b\in E_2,c \in E_3$. The goal is to select the maximum number of triplets such as no element belongs to more than one triplet. $3\mathrm{D}$-$2$-$\mathrm{Matching}$ (3D-2-M) is the special case of $3\mathrm{D}$-$\mathrm{Matching}$ where each element $e E = \in E_1 \cup E_2 \cup E_3$ belongs in exactly $2$ triplets.
\end{definition}
\begin{theorem}[\cite{CC06}]\label{t:prev}
  There is no polynomial-time approximation algorithm for $\mathrm{3D}$-2-$\mathrm{M}$ with approximation ratio $\rho \geq 94/95$ unless $\mathrm{P = NP}$.
\end{theorem}
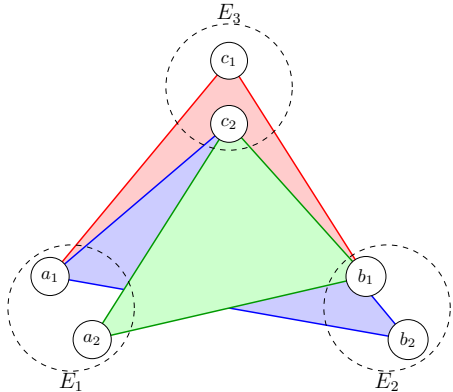
\begin{figure}[t]
  \centering
  \scalebox{0.7}{
	\begin{tikzpicture}[
      setlabel/.style={draw=none, font=\large},
      triplet1/.style={draw=red, fill=red!20, thick},
      triplet2/.style={draw=blue, fill=blue!20, thick},
      triplet3/.style={draw=green!60!black, fill=green!20, thick},
      nodeStyle/.style={circle, draw, fill=white, minimum size=7mm}
      ]

      \coordinate (E1) at (0,0);
      \coordinate (E2) at (6,0);
      \coordinate (E3) at (3,4);

      \coordinate (a1) at (-0.4,0.6);
      \coordinate (a2) at (0.4,-0.6);

      \coordinate (b1) at (5.6,0.6);
      \coordinate (b2) at (6.4,-0.6);

      \coordinate (c1) at (3,4.7);
      \coordinate (c2) at (3,3.5);

      \draw[triplet1] (a1) -- (b1) -- (c1) -- (a1);

      \draw[triplet2] (a1) -- (b2) -- (c2) -- (a1);

      \draw[triplet3] (a2) -- (b1) -- (c2) -- (a2);

      \node[nodeStyle] at (a1) {$a_1$};
      \node[nodeStyle] at (a2) {$a_2$};

      \node[nodeStyle] at (b1) {$b_1$};
      \node[nodeStyle] at (b2) {$b_2$};

      \node[nodeStyle] at (c1) {$c_1$};
      \node[nodeStyle] at (c2) {$c_2$};

      \node[setlabel] at (0,-1.4) {$E_1$};
      \node[setlabel] at (6,-1.4) {$E_2$};
      \node[setlabel] at (3,5.6) {$E_3$};

      \draw[dashed] (0,0) circle (1.2);
      \draw[dashed] (6,0) circle (1.2);
      \draw[dashed] (3,4.2) circle (1.2);

	\end{tikzpicture}}
  \caption{\label{fig:matching-example}
    A \textsc{3D-Matching} instance with
    \(E_1=\{a_1,a_2\}\), \(E_2=\{b_1,b_2\}\), and \(E_3=\{c_1,c_2\}\)
    arranged at the vertices of an abstract triangle.
    Each filled colored triangle represents a triplet in
    \(S=\{(a_1,b_1,c_1),(a_1,b_2,c_2),(a_2,b_1,c_2)\}\).
    All three edges of every triplet are explicitly drawn.
    A valid \textsc{3D-Matching} corresponds to selecting vertex-disjoint triangles;
    here, the maximum matching has size~1.
  }
\end{figure}

\subsection{Constructing the RMFUP instance}
In this section, we present how to construct an instance of RMFUP given a 3D-2-M instance. In particular, given the set of elements $E_1,E_2,E_3$ and triplets $S$ of a $3\mathrm{D}$-2-$\mathrm{M}$ instance, we construct the following RMFUP instance with $n := |E_1| + |E_2| + |E_3| + m$ buyers and $m:= |S|$ goods as follows:
\begin{itemize}
  \item For each element $e\in E_1\cup E_2 \cup E_3$, consider a buyer $e$ with budget $B_e = 1/3$.
        \smallskip

  \item For each triplet $s = (e_1,e_2,e_3)\in S$, consider a good $s \in [m]$ and let the $v_{e_1s} = v_{e_2s} = v_{e_3s} = 1$.
        \smallskip

  \item Each buyer $e$ has valuation $v_{es} = 0$ if $e \notin s$.
        \smallskip

  \item Each good/triplet $s \in [m]$, admits a special buyer $s$ with $v_{ss} = 2/3$ and $B_{s} = 2/3$. At the same time, $v_{ss'} = 0$ for all goods $s' \neq s$.
\end{itemize}
\noindent In Figure~\ref{fig:example} we illustrate the RMFUP instance for the \textsc{3D-2-M} instance shown of Figure~\ref{fig:matching-example}.

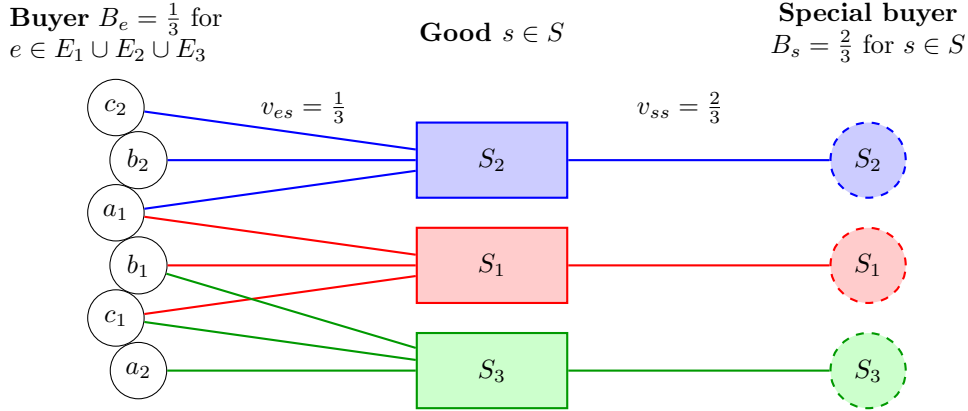
\begin{figure}[t]
  \begin{tikzpicture}[
    buyer/.style={circle, draw, minimum size=0.75cm},
    specialbuyer/.style={circle, draw,dashed, minimum size=1cm, },
    item/.style={rectangle, draw, minimum width=2.0cm, minimum height=1cm},
    ]
    \def\buyerX{0}
    \def\itemX{5}
    \def\specialBuyerX{10}

    \node[align=left] at (\buyerX, 4.5) {\textbf{Buyer} $B_e = \frac{1}{3}$ for \\ $e \in E_1 \cup E_2 \cup E_3$ };
    \node at (\itemX, 4.5) {\textbf{Good $s \in S$}};
    \node[align=center] at (\specialBuyerX, 4.5) {\textbf{Special buyer} \\  $B_s = \frac{2}{3}$ for $s \in S$};

    \node at (7.5, 3.5) {$v_{s s} = \frac{2}{3}$};
    \node at (2.5, 3.5) {$v_{e s} = \frac{1}{3}$};

    \node[buyer] (A1) at (\buyerX, 2.1) {$a_1$};
    \node[buyer] (B1) at (\buyerX+0.3, 1.4) {$b_1$};
    \node[buyer] (C1) at (\buyerX, 0.7) {$c_1$};
    \node[buyer] (A2) at (\buyerX+0.3, 0.0) {$a_2$};
    \node[buyer] (B2) at (\buyerX+0.3, 2.8) {$b_2$};
    \node[buyer] (C2) at (\buyerX, 3.5) {$c_2$};

    \node[item,fill=red!20  ,thick,draw=red  ] (S1) at (\itemX,1.4) {$S_1$};
    \node[item,fill=blue!20 ,thick,draw=blue ] (S2) at (\itemX,2.8) {$S_2$};
    \node[item,fill=green!20,thick,draw=green!60!black] (S3) at (\itemX,0.0) {$S_3$};

    \node[specialbuyer,fill=red!20  ,thick,draw=red  ] (aS1) at (\specialBuyerX,1.4) {$S_1$};
    \node[specialbuyer,fill=blue!20 ,thick,draw=blue ] (aS2) at (\specialBuyerX,2.8) {$S_2$};
    \node[specialbuyer,fill=green!20,thick,draw=green!60!black] (aS3) at (\specialBuyerX,0.0) {$S_3$};

    \draw[thick,fill=red!20  ,thick,draw=red  ] (A1) -- (S1);
    \draw[thick,fill=red!20  ,thick,draw=red  ] (B1) -- (S1);
    \draw[thick,fill=red!20  ,thick,draw=red  ] (C1) -- (S1);

    \draw[thick,fill=blue!20 ,thick,draw=blue ] (A1) -- (S2);
    \draw[thick,fill=blue!20 ,thick,draw=blue ] (B2) -- (S2);
    \draw[thick,fill=blue!20 ,thick,draw=blue ] (C2) -- (S2);

    \draw[thick,fill=green!20,thick,draw=green!60!black] (A2) -- (S3);
    \draw[thick,fill=green!20,thick,draw=green!60!black] (B1) -- (S3);
    \draw[thick,fill=green!20,thick,draw=green!60!black] (C1) -- (S3);

    \draw[thick,fill=red!20  ,thick,draw=red  ] (aS1) -- (S1);
    \draw[thick,fill=blue!20 ,thick,draw=blue ] (aS2) -- (S2);
    \draw[thick,fill=green!20,thick,draw=green!60!black] (aS3) -- (S3);
  \end{tikzpicture}
  \caption{Visualization of the conversion of example from Figure~\ref{fig:matching-example}. Recall that it is defined by \(E_1=\{a_1,a_2\}\), \(E_2=\{b_1,b_2\}\), and \(E_3=\{c_1,c_2\}\) and triplets \(S=\{(a_1,b_1,c_1),(a_1,b_2,c_2),(a_2,b_1,c_2)\}\). Notice that element buyers \(e\) have budget \(B_e = \frac{1}{3}\) and special buyers \(s\) have budget \(B_s = \frac{2}{3}\). The colors of the special buyers and goods correspond to the color of the triplet from Figure~\ref{fig:matching-example}.}
  \label{fig:example}\label{f:example}
\end{figure}

Up next, we show how an approximate solution for an approximate solution for RMFUP can be transformed in polynomial-time to an approximate solution for 3D-2-M. We remark that a feasible solution for RMFUP satisfies the constraints outlined in Subsection~\ref{ss:RMFUP} and is completely determined by the payment and price vector $(\mathbf{b},\mathbf{p})$ (the allocation of good $s$ to buyer $i$, can be simply determined as $x_{is} = b_{is}/p_s$).

\smallskip
Before converting a solution $(\mathbf{b},\mathbf{p})$ of RMFUP to a solution of 3D-2-M, we first \textit{round} $(\mathbf{b},\mathbf{p})$ to another solution $(\mathbf{\hat{b}},\mathbf{\hat{p}})$ of RMFUP so that the properties outlined in Lemma~\ref{l:1} are satisfied.
\begin{lemma}\label{l:1}
  There exists a polynomial-time algorithm (Algorithm~\ref{alg:round}, see Appendix~\ref{app:red}) that transforms a solution $(\mathbf{b},\mathbf{p})$ to a solution $(\mathbf{\hat{b}},\mathbf{\hat{p}})$ with the following properties:
  \begin{enumerate}
    \item $p_s = 2/3$ or $p_s = 1$ for all goods $s \in [m]$.
          \smallskip

    \item If $p_s = 2/3$ then $s$ is completely sold to the special buyer $s$, $x_{ss} = 1$.
          \smallskip

    \item Let $s = (e_1,e_2,e_2)$ with price $p_s = 1$, then it is completely sold to $e_1,e_2,e_2$ ($x_{e_1} + x_{e_2} + x_{e_3} = 1)$.
    \item The revenue of the updated solution does not decrease, \[\sum_{i = 1}^n \sum_{j=1}^m \hat{b}_{is} \geq  \sum_{i = 1}^n \sum_{j=1}^m b_{is}.\]
  \end{enumerate}
\end{lemma}
Due to space constraints, Algorithm~\ref{alg:round} and the proof of Lemma~\ref{l:1} are deferred to Appendix~\ref{app:red}. The main idea is the following. Since buyers’ valuations take values only in \(\{0, \tfrac{2}{3}, 1\}\), we can increase the price of any good with \(p_s < \tfrac{2}{3}\) to \(\hat{p}_s = \tfrac{2}{3}\), and any good with \(p_s \in (\tfrac{2}{3}, 1]\) to \(\hat{p}_s = 1\), while appropriately reducing the allocated quantity. These transformations do not violate either individual rationality or budget constraints. Moreover, when \(\hat{p}_s = \tfrac{2}{3}\), the good can be entirely allocated to the special buyer without affecting the total revenue.
\smallskip

\noindent With a slight abuse of notation in the rest of the paper, we assume that $(\mathbf{b},\mathbf{p})$ satisfies the properties of Lemma\ref{l:1}. With this at hand we can easily derive feasible solution for 3D-2-M.
\begin{claim}
  All good/triplets $s \in [m]$ with $p_s = 1$ form a feasible solution for $\mathrm{3DM}$.
\end{claim}
\begin{proof}
  Property~3 guarantees that a good/triplet $s \in [m]$ with $p_s = 1$ is completely sold to exactly $3$ buyers/elements. This means that the seller receives exactly $1/3$ from each one of them (their budget is exactly $1/3$). As a result, none of the 3 buyers can contribute to any other good with price $p_s = 1$. Also notice that due to property~$2$ all goods $s \in [m]$ with $p_s = 2/3$ are completely sold to the special buyers. This means that no buyer/element participates in more than one good/triplet.
\end{proof}
\noindent Lemma~\ref{l:2} is the main technical lemma of the section and establishes that the set of goods/triplets with $p_s = 1$ forms an approximate solution for the 3D-2-M instance.
\begin{lemma}\label{l:2}
  Let $(\mathbf{b},\mathbf{p})$ be an $\rho$-approximate solution for RMFUP, then the good/triplets $s \in [m]$ with $p_s = 1$ form an $(9 \rho - 8)$-approximate solution for the $3\mathrm{D}$-2-$\mathrm{M}$ instance.
\end{lemma}
\begin{proof}
  Let $M^\star$ denote the optimal solution of 3D-2-M and consider the following solution $(\mathbf{b^\star},\mathbf{p^\star})$ for RMFUP,
  \[p_s =
    \begin{cases}
      1,   & \text{for each good } s\in M^\star \\
      2/3, & \text{otherwise}
    \end{cases}\]
  If $p^\star_s = 1$ the good is completely sold to the respective $3$ buyers/elements and if $p^\star_s = 2/3$ the good is completely sold to the special buyer. Since in $M^\star$ no element participates in two different triplets, we are ensured that $(\mathbf{b^\star},\mathbf{p^\star})$ is budget feasible.
  Since $(\mathbf{b},\mathbf{p})$ is a $\rho$-approximate optimal solution for RMFUP,
  \[ \sum_{i \in [n]} \sum_{j \in [n]} b_{ij} \geq \rho \cdot  \sum_{i \in [n]} \sum_{j \in [n]} b^\star_{ij} = \rho \cdot \left( \frac{2}{3} \left( m - |M^\star|\right) + |M^\star| \right).\]
  Let $\hat{M}:= \{s \in [m]~:~ p_s = 1\}$ by Properties~(2)~and~(3) of Lemma~\ref{l:1} we get that
  \[\sum_{i \in [n]} \sum_{j \in [n]} b_{ij} =  \left( \frac{2}{3} \left( m - |\hat{M}|\right) + |\hat{M}| \right) \geq \rho \cdot \left( \frac{2}{3} \left( m - |M^\star|\right) + |M^\star| \right)\]
  and thus we conclude that
  \begin{equation}\label{eq:98}
    |\hat{M}| \geq 2m \cdot( \rho - 1) + \rho \cdot M^\star.
  \end{equation}
  \noindent Dividing $\ref{eq:98}$ by $|M^\star|$ we get that
  \begin{equation}\label{eq:198}
    \frac{|\hat{M}|}{|M^\star|} \geq \rho - 2 ( 1 - \rho)\cdot \frac{m}{|M^\star|}.
  \end{equation}
  The latter means that we need to provide a upper bound on $\frac{m}{|M^\star|}$ that in 3DM instances can be arbitrarily large. However for 3D-2-M instances  is upper bounded by $4$.
  \begin{lemma}\label{l:matching}
    For any instance of 3D-2-M we have that $m \leq 4 \cdot |M^\star|$ where $m$ is the number of triplets and $M^\star$ the optimal solution.
  \end{lemma}
  \begin{proof}
    Consider the dependency graph $G(V,E)$ where each triplet $s \in [m]$ corresponds to a node $v_s \in V$ and two nodes $v,v' \in V$ are connected if and only if they share an element. Since each element $e \in E$ appears in at most $2$ triplets, the degree of each node $v_s \in V$ is at most $3$. Now consider a maximal independent set in the dependency graph, denoted as $M$. Clearly $|M| \leq |M^\star|$. Since $M$ is a maximal independent set $N(M) \cup M = V$. As a result, $m = |V| \leq |M| + |N(M)| \leq 4 |M| \leq 4 |M^\star|$.
  \end{proof}
  \noindent Lemma~\ref{l:2} follows by considering $m / |M^\star| \leq 4$ in Equation~\ref{eq:198}.
\end{proof}
\noindent By requiring $9 \rho - 8 < 94/95$ we obtain that $\rho < 854/855$. 
\begin{theorem}\label{t:inappr}
  There is no polynomial-time approximation algorithm for RMFUP with approximation ratio $\rho < 854/855$ unless P = NP.
\end{theorem}

\section{Revenue Guarantees of First-Price Pacing Equilibrium}\label{s:revenue_linear}

In this section, we establish revenue guarantees for the \textit{First-Price Pacing Equilibrium} (FPPE). As discussed earlier, FPPE induces fixed unit prices and can therefore be viewed as an approximation algorithm for RMFUP. Our main result shows that FPPE achieves at least $50\%$ of the optimal revenue, even when compared to mechanisms with variable unit prices.

\begin{theorem}\label{t:revenue}
  $\mathrm{FPPE} \geq \frac{1}{2} \cdot \mathrm{RMVUP}$.\footnote{With a slight abuse of notation, we denote the respective revenues by FPPE and RMVUP.}
\end{theorem}

Theorem~\ref{t:revenue} establishes that, despite relying on fixed unit prices, the revenue generated by FPPE is at least one half of the optimal revenue achievable with variable unit prices. This result directly implies that FPPE constitutes a $\frac{1}{2}$-approximation algorithm for RMFUP, which is an APX-hard problem. Interestingly, we show the revenue guarantee of Theorem~\ref{t:revenue} is optimal with respect to fixed unit-price mechanisms. In particular, we construct instances for which $\mathrm{RMFUP} = \frac{1}{2} \cdot \mathrm{RMVUP}$.

\begin{restatable}{theorem}{RMFUPBound}\label{t:lower_bound1}
  In the worst-case the ratio $\frac{\mathrm{RMFUP}}{\mathrm{RMVUP}} = \frac{1}{2}$.
\end{restatable}

\noindent
Due to space constraints, the proof of Theorem~\ref{t:lower_bound1} is deferred to Appendix~\ref{app:2}. In the remainder of this section, we prove Theorem~\ref{t:revenue}. With a slight abuse of notation, we denote the revenue generated by FPPE simply by FPPE respectively for RMVUP.

The first step of our approach consists of bounding the Liquid Social Welfare of the FPPE with respect to the optimal Liquid Social Welfare.
\begin{definition}\label{d:lw}
  The \textit{Liquid Social Welfare} of an allocation $\mathbf{x}$, is defined as
  \[ \mathrm{LW}(\mathbf{x}):= \sum_{i = 1}^n \min \left(\sum_{j=1}^m v_{ij}\cdot x_{ij} , B_i \right)\]
\end{definition}
\noindent The Liquid Social Welfare is a very important notion since $\mathrm{LW}(\mathbf{x})$ it serves as an upper bound on the revenue that can be attained through an allocation $\mathbf{x}$ under individual rationality and budget constraints.

In Lemma~\ref{t:POA_LW} we establish the fact that Liquid Social Welfare produced by FPPE is at least $1/2$ of the optimal Liquid Social Welfare.
\begin{restatable}{lemma}{lemPOALW}\label{t:POA_LW}
  Let $(\mathbf{x},\mathbf{b},\mathbf{p}, \mathbf{\alpha})$ denote an FPPE of a market and $\mathbf{o^\star}$ the allocation maximizing Liquid Social Welfare. Then, $\mathrm{LW}(\mathbf{x}) \geq \frac{1}{2}\cdot \mathrm{LW}(\mathbf{o^\star})$.
\end{restatable}
\begin{proof}
  To simplify notation let $v_i(x_i) := \sum_{j\in [n]} v_{ij} \cdot x_{ij}$. Without loss of generality, we assume that $v_i(o^\star_i)\leq B_i$ for every buyer $i\in [n]$ since otherwise we can simply decrease an $o^\star_{ij} > 0$. Let $V\subseteq [n]$ denote the subset of buyers satisfying $B_i\geq v_i(x_i)$. Then by Definition~\ref{d:lw} we have that
  \begin{align}\label{eq:LW-xb}
    \lw(\x) & = \sum_{i\in [n]}{\min\left\{B_i,v_i(x_i)\right\}}=\sum_{i\in V}{B_i}+\sum_{i\in [n]\setminus V}{v_i(x_i)}
  \end{align}
  and
  \begin{align}\label{eq:LW-xstarb}
    \lw(\mathrm{o}^\star) & =\sum_{i\in [n]} \min\{B_i,{v_i(o^\star_i)}\} \leq \sum_{i\in V}{B_i}+\sum_{i\in [n]\setminus V}{v_i(o^\star_i)}
  \end{align}
  \noindent Given a buyer $i\in [n]\setminus V$. We distinguish between two cases. If $v_i(o^\star_i)< \sum_{j\in [m]}{o^\star_{ij}\cdot p_j}$ then
  \begin{align}\label{eq:utility-positiveb}
    v_i(x_i)-\sum_{j\in [m]}{x_{ij}\cdot p_j} & \geq 0 >v_i(o^\star_i)-\sum_{j\in [m]}{o^\star_{ij}\cdot p_j}.
  \end{align}
  where the first inequality comes from the utility maximization property of FPPE. If $v_i(o^\star_i)\geq \sum_{j\in [m]}{o^\star_{ij}\cdot p_j}$ then our assumption about the allocation $\mathrm{o}^\star$ yields
  \begin{equation}\label{eq:100a}
    \sum_{j\in [m]}{o^\star_{ij}\cdot p_j} \leq v_i(o^\star_i) \leq B_i
  \end{equation}

  \noindent The latter means that the allocation $\mathbf{o}^\star$ satisfies the budget constraint for buyer $i \in [n]$ and thus, the utility maximization property of FPPE (see Section~\ref{s:prelim}) implies that
  \begin{align}\label{eq:utility-maximizinga}
    v_i(x_i)-\sum_{j\in [m]}{x_{ij}\cdot p_j} & \geq v_i(o^\star_i)-\sum_{j\in [m]}{o^\star_{ij}\cdot p_j}.
  \end{align}
  As a result, for any buyer $i \in [n]$ we get that
  \begin{align}\label{eq:competitive-equilibrium-conditiona}
    v_i(o^\star_i) & \leq v_i(x_i)+\sum_{j\in [m]}{p_j\cdot (o^\star_{ij}-x_{ij})},
  \end{align}
  Then Equations~(\ref{eq:competitive-equilibrium-conditiona}) and~(\ref{eq:LW-xstarb}) imply that
  \begin{align}\nonumber
    \lw(\mathrm{o}^\star) & \leq \sum_{i\in V}{B_i}+\sum_{i\in [n]\setminus V}{v_i(x_i)}+\sum_{i\in [n]\setminus V}{\sum_{j\in [m]}{p_j\cdot (o^\star_{ij}-x_{ij})}} \\\label{eq:bound-liquid-plus-pricesb}
             & = \lw(\x)+\sum_{i\in [n]\setminus V}{\sum_{j\in [m]}{p_j\cdot (o^\star_{ij}-x_{ij})}}.
  \end{align}
  where the last equality is due to Equation~(\ref{eq:LW-xb}). Finally, we show that
  \begin{align}\label{eq:bound-for-pricesa}
    \sum_{i\in [n]\setminus V}{\sum_{j\in [m]}{p_j\cdot (o^\star_{ij}-x_{ij})}} & \leq \lw(\x).
  \end{align}
  The proof of the lemma then follows by Equation~(\ref{eq:bound-liquid-plus-pricesb}). We have that
  \begin{align*}
    \sum_{i\in [n]\setminus V}{\sum_{j\in [m]}{p_j\cdot (o^\star_{ij}-x_{ij})}} & = \sum_{j\in [m]}{p_j \sum_{i\in [n]\setminus V}{(o^\star_{ij}-x_{ij})}}  \leq \sum_{j\in [m]}{p_j\cdot \left(1- \sum_{i\in [n]\setminus V}{x_{ij}}\right)} \\
                                                      & = \sum_{j\in [m]}{p_j\sum_{i\in V}{x_{ij}}}
                                                        = \sum_{i\in V}{\sum_{j\in [m]}{p_j\cdot x_{ij}}} \leq \sum_{i\in V}{B_i}
                                                        \leq \lw(\x),
  \end{align*}
  The first inequality is due to $\sum_{j =1}^mo^\star_{ij} \leq 1$. The second inequality is due to the budget constraint for buyer $i\in V$, $\sum_{j\in [m]}{p_j\cdot x_{ij}}\leq B_i$. The third inequality is due to Equation~(\ref{eq:LW-xb}). The second equality is due to Property~$(5)$ of FPPE, if $p_j>0$ for good $j\in [m]$ then $\sum_{i\in [n]}{x_{ij}}=1$.
\end{proof}

In the rest of section we consider $(\mathbf{x},\mathbf{b},\mathbf{p}, \mathbf{\alpha})$ to be the FPPE of the market, $(\mathbf{x^\star},\mathbf{b^\star})$ to be RMVUP, and provide the proof of Theorem~\ref{t:revenue}. Lemma~\ref{t:POA_LW} establishes the fact that the Liquid Social Welfare of the FPPE is at least $1/2$ of the RMVUP. This is because
\begin{eqnarray*}
  \mathrm{LW}(x) &\geq& \frac{1}{2} \cdot \mathrm{LW}(x^\star)= \frac{1}{2}\cdot \sum_{i=1}^n \min \left( \sum_{j=1}^m v_{ij}\cdot x^\star_{ij} , B_i \right)\\
        &\geq& \frac{1}{2} \sum_{i=1}^n \sum_{j \in [m]}b^\star_{ij} = \frac{1}{2}\cdot \mathrm{RMVUP}.
\end{eqnarray*}

\noindent The cornerstone idea of our proof is that FPPE has the interesting property that its Liquid Social Welfare equals exactly its revenue. In particular,
\[\mathrm{FPPE}:=\sum_{i=1}^n\sum_{j=1}^m b_{ij} = \mathrm{LW}(x)\]

In order to establish the latter, consider the set of buyers $A := \{\sum_{j \in [m]} b_{ij} < B_i \}$. By Property~(6) of FPPE we are ensured that each buyer $i\in A$ admits pacing multiplier $\alpha_i = 1$. By Property~$(3)$ we know that for any good $j \in [m]$ such that $x_{ij} > 0$, $p_j = \alpha_i \cdot v_{ij}$ meaning that
\[ \sum_{j \in [m]} b_{ij} = \sum_{j\in [m]} v_{ij}\cdot x_{ij}\text{ for all buyers } i \in A\]
Now consider the revenue of the FPPE,
\begin{eqnarray*}
  \mathrm{FPPE} &:=& \sum_{i=1}^n \sum_{j \in [m]}b_{ij} = \sum_{i \in A} \sum_{j \in [m]}p_j \cdot x_{ij}  + \sum_{i \notin A} B_i\\
       &=& \sum_{i \in A} \sum_{j \in [m]}p_{j}\cdot x_{ij} + \sum_{i
           \notin A} B_i = \sum_{i \in [n]} \min\left( \sum_{j \in [m]}v_{ij}\cdot x_{ij} , B_i \right) = \mathrm{LW}(\mathbf{x})
\end{eqnarray*}

\section{First-Price Pacing Equilibrium as an Online Algorithm}\label{s:online}

In online advertising markets, campaigns arrive over time and impressions are generated on a daily basis. This motivates an online formulation of the revenue maximization problem, in which the platform must make allocation and pricing decisions sequentially, without knowledge of future demand. At each day, the platform observes the currently active advertisers and allocates that day’s impressions among them, subject to individual rationality and feasibility and budget constraints.

In the section, we propose the following online problem to capture the setting above:

\begin{itemize}
  \item There exists a finite time horizon $[1,T]$ and $m$ divisible goods corresponding to the possible keywords. At the beginning of each round $ t \in [T]$, each good is completely renewed while it expires at the end of a round.
        \smallskip

  \item There is a sequence of $n$ buyers arriving online and each buyer \(i \in [n]\) is associated with an interval \([s_i, t_i] \subseteq [1,T]\) denoting its \textit{arrival} and \textit{departure time}. Each buyer \(i \in [n]\) has valuations over the goods \(v_{i1}, \dots, v_{im}\) that remain constant across $[s_i ,t_i]$ and a overall budget $B_i \geq 0$ for the whole interval $[s_i,t_i]$. Finally buyer $i \in [n]$ is revealed to the system at time $s_i \in [T]$.
	    \smallskip
\end{itemize}
\smallskip

At each round \(t \geq 1\), an \textit{online algorithm} with \textit{fixed unit prices} needs to select unit prices \(p_j^t \ge 0\) for each good \(j \in [m]\) as well as allocate  the goods to the active buyers $S(t):=\{ i \in [n] \mid t \in [s_i, t_i] \}$,
\[
  \sum_{i \in S(t)} x_{ij}^t \le 1 \qquad \text{for all } j \in [m].
\]
The extracted revenue at round $t \geq 1$ equals $
R_t:= \sum_{i \in S(t)} \sum_{j\in [m]} p_j^t \cdot  x_{ij}^t.$
An \textit{online algorithm} must also satisfy the \textit{individual rationality} and \textit{budget feasibility} constraints. For each round $t \geq 1$
\[
  p_j^t \cdot  x_{ij}^t \leq v_{ij} \cdot  x_{ij}^t\text{   }~~~~~~~~\forall i \in S(t),j \in [m]
\]
and for each buyer $i \in [n]$,
$\sum_{t = 1}^T \sum_{j \in [m]} p_j^t \cdot  x_{ij}^t \leq B_i.$
\smallskip
\smallskip

Given an offline access to the sequence of buyers $i \in [n]$, the optimal revenue with respect to \textit{variable unit prices} by reducing it to an \textit{offline instance} and solving the linear program of RMVUP. In particular, one can treat each good $j \in [m]$ at round $t \geq 1$, as a separate good $j_t$ and for each buyer $i \in [n]$ consider the valuations,
\[v_{ij}^t =
  \begin{cases}
    v_{ij}, & \text{if } t \in [s_t,t_i] \\
    0,      & \text{ otherwise}.
  \end{cases}
\]
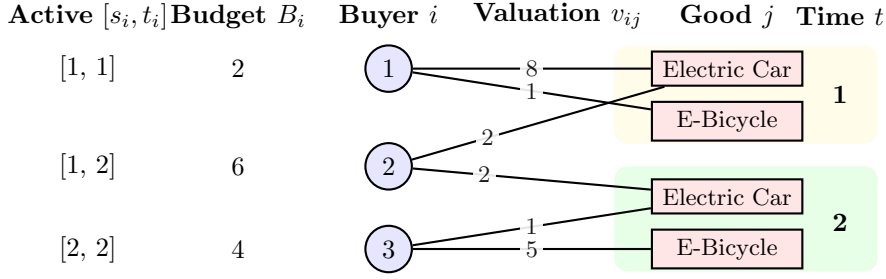
\begin{figure}[t]
  \begin{tikzpicture}[
    buyer/.style={circle, draw, minimum size=\baselineskip+0.3, fill=blue!10, thick},
    item/.style={rectangle, draw, minimum width=2cm, minimum height=\baselineskip+0.3, fill=red!10, thick},
    valuation/.style={font=\small, fill=white, inner sep=2pt, opacity=0.9},
    ]

    \def\TimeWindowX{-4.0}
    \def\BudgetX{-2}
    \def\BuyerX{0}
    \def\ValuationX{2.25}
    \def\ItemX{4.5}
    \def\TimeX{6.0}
    \def\NodeFontSize{\normalsize}

    \fill[yellow!10, rounded corners=5pt] (\ItemX-1.5, 4.6) rectangle (\TimeX+0.5, 3.3);

    \fill[green!10, rounded corners=5pt] (\ItemX-1.5, 3.0) rectangle (\TimeX+0.5, 1.6);

    \node [align=center] at (\TimeWindowX, 5.0) {\textbf{Active} $[s_i,t_i]$};
    \node at (\BudgetX, 5.0) {\textbf{Budget $B_i$}};
    \node at (\BuyerX, 5.0) {\textbf{Buyer $i$}};
    \node at (\ValuationX, 5.0) {\textbf{Valuation $v_{ij}$}};
    \node at (\ItemX, 5.0) {\textbf{Good $j$}};
    \node at (\TimeX, 5.0) {\textbf{Time $t$}};

    \node[anchor=center, font=\NodeFontSize] at (\TimeWindowX, 4.3) {[1, 1]};
    \node[buyer, font=\NodeFontSize] (A1) at (\BuyerX, 4.3) {1};
    \node[anchor=center, font=\NodeFontSize] at (\BudgetX, 4.3) {2};

    \node[anchor=center, font=\NodeFontSize] at (\TimeWindowX, 3.0) {[1, 2]};
    \node[buyer, font=\NodeFontSize] (A2) at (\BuyerX, 3) {2};
    \node[anchor=center, font=\NodeFontSize] at (\BudgetX, 3) {6};

    \node[anchor=center, font=\NodeFontSize] at (\TimeWindowX, 1.9) {[2, 2]};
    \node[buyer, font=\NodeFontSize] (A3) at (\BuyerX, 1.9) {3};
    \node[anchor=center, font=\NodeFontSize] at (\BudgetX, 1.9) {4};

    \node[item, font=\small] (I1) at (\ItemX, 4.3) {Electric Car};
    \node[item, font=\small] (I2) at (\ItemX, 3.6) {E-Bicycle};
    \node at (\TimeX, 3.95) {\textbf{1}};

    \node[item, font=\small] (I1_t2) at (\ItemX, 2.6) {Electric Car};
    \node[item, font=\small] (I2_t2) at (\ItemX, 1.9) {E-Bicycle};
    \node at (\TimeX, 2.25) {\textbf{2}};

    \draw[thick] (A1) -- node[valuation] {8} (I1);
    \draw[thick] (A1) -- node[valuation] {1} (I2);

    \draw[thick] (A2) -- node[valuation, pos=0.3] {2} (I1);
    \draw[thick] (A2) -- node[valuation, pos=0.3] {2} (I1_t2);

    \draw[thick] (A3) -- node[valuation] {1} (I1_t2);
    \draw[thick] (A3) -- node[valuation] {5} (I2_t2);

  \end{tikzpicture}
  \caption{\label{fig:online-fppe} Diagram showing an example of an online FPPE problem. Notice that buyer 1 and 2 are fully revealed at time-step 1 and buyer 3 is only revealed to at time-step 3. }
\end{figure}
With a slight abuse of notation we denote as RMVUP the total revenue that can be extracted by the sequence of buyers with variable unit prices. The performance of online algorithm is captured through the notion of \textit{competitive ratio}, defined as $\sum_{t=1}^T R_t / \mathrm{RMVUP}$.
\smallskip

In Theorem~\ref{t:any-competitive}, we establish an upper bound on the competitive ratio of any online algorithm.
\begin{theorem}\label{t:any-competitive}
  The competitive ratio of any algorithm is at most \( \frac{1}{4} \left(2 + \sqrt{2}\right) \approx 0.85\).
\end{theorem} 
\noindent The proof of Theorem~\ref{t:any-competitive} is deferred to Appendix~\ref{a:online}.

\subsection{FPPE as an Online Algorithm}\label{s:online-fppe}
The main result of the section consists of establishing that greedily applying FPPE at each round, leads to an $1/4$-competitive algorithm.

\begin{algorithm}[h]
  \caption{Online First-Price Pacing Equilibrium}
  \label{alg:1}
  \DontPrintSemicolon
  \ForEach{round $t \geq 1$}{
    \smallskip

    The set of buyers $A(t) := \{ i \in [n]:s_i = t \}$ arrive 
    \smallskip

    \ForEach {buyer $i \in A(t)$}{
      $B^t_i \gets B_i$
    }
    \smallskip
    Let the set of \textit{active} buyers $S(t) := \{ i \in [n] : t \in [s_i, t_i] \}$ at round $t \in [n]$.
    \smallskip

    Compute an FPPE $(\mathbf{x}^t, \mathbf{b}^t, \boldsymbol{p}^t,\mathbf{\alpha}^t)$ for buyers $i\in S(t)$ with budgets $B_i^t$.
    \smallskip

    Update the budget of each active buyer $ i \in S(t)$,
    \[B_i^{t+1} \gets B_i^t - \sum_{j=1}^m p_j^t \cdot x_{ij}^t\]
  }
\end{algorithm}

Algorithm~\ref{alg:1} simply computes, at each round $t \geq 1$, an FPPE for the set of active buyers $S(t):= \{i \in [n]: t \in [s_i,t_i]\}$ and then updates the budgets according to what each buyer has spent. In Theorem~\ref{t:online} we establish that this straightforward generalization of FPPE in the online setting produces a $1/4$-competitive algorithm.
\begin{theorem}\label{t:online}
  Online FPPE is an $1/4$-competitive algorithm.
\end{theorem}

\begin{remark}
  We remark that Algorithm~\ref{alg:1} set \textit{fixed unit-prices} (w.r.t. to buyers) at any round $t \geq 1$. However the price of a good $j \in [m]$ may differ from round to round, $p_j^t \neq p_j^{t'}$. However, having different unit prices across different time-steps is very reasonable, since the demand of a good can be subject to time, which naturally lead to price variations. For example a keyword for "airline tickets" may increase across summer months or around the Christmas holidays.
\end{remark}
\noindent In the rest of the section we provide the main steps for establishing Theorem~\ref{t:online}. With a slight abuse of notation, we denote with FPPE the overall revenue produced by Algorithm~\ref{alg:1} across all the $T$ rounds.

Our overall goal is to establish that
\[\mathrm{FPPE}\geq \frac{1}{4} \cdot \mathrm{RMVUP}.\]
In the rest of the section we denote with $b_i^t := \sum_{j \in [m]}b_{ij}^t$, the overall spending of buyer $i \in [n]$ at round $t \in [T]$ in FPPE and similarly let $o_i^t := \sum_{j \in [m]}o_{ij}^t$ denote the spending of that agent in the RMVUP\footnote{Specifically \(o_{ij}^t\) is the price of good \(j\) when sold to buyer \(i\).}. Under this notation, our goal takes the following form,
\[\sum_{t=1}^T \sum_{i \in S(t)} b_i^t \geq \frac{1}{4} \cdot \sum_{t=1}^T \sum_{i \in S(t)} o_i^t.\]

In order to establish the claim above, we introduce an \textit{intermediate} solution that is constructed by combining the outputs of Algorithm~\ref{alg:1} with the RMVUP. 

\begin{definition}[intermediate solution]\label{d:hat}
  For each round $t \geq 1$ and buyer $i \in S(t)$, let the budgets $\hat{B}^t_i = \max( o_i^t, B_i^t)$. The intermediate solution $(\hat{x}^t,\hat{b}^t,\hat{p}^t)$ is defined as as the FPPE computed with respect to budgets $\hat{B}^t_i$ for the buyers $i \in S(t)$. Its total revenue equals,
  \[\sum_{t=1}^T \sum_{i\in S(t)}\hat{b}^t_{i},\]
  where $\hat{b}^t_{i}:= \sum_{j \in [m]}\hat{b}^t_{ij}$ is the spending of buyer $i \in [n]$ at round $t \in [T]$ in the intermediate solution.
\end{definition}
\noindent We remark that the intermediate solution is not necessarily budget feasible. In particular, it can be the case that $\sum_{t = s_i}^{t_i} \hat{b}^t_i \geq B_i$ for some buyer $i \in [n]$. However this should not raise any concern since the intermediate solution is only used for the sake of the analysis.

The reason that the intermediate solution is crucial for the analysis is that it permits us establish a \textit{time-wise comparison} with RMVUP,
\begin{equation}\label{eq:17}
  \underbrace{\sum_{i\in S(t)} \hat{b}^t_{i} \geq \frac{1}{2} \sum_{i\in S(t)}o^t_{i}~~~~~~\text{       for each }t \geq 1}_{\text{time-wise comparison}}
\end{equation}
\noindent while at the same time permits us a \textit{buyer-wise comparison} between Algorithm~\ref{alg:1} and the intermediate solution,
\begin{equation}\label{eq:19}
  \underbrace{\sum_{t=s_i}^{t_i} b^t_{i} \geq \frac{1}{2} \sum_{t=s_i}^{t_i} \hat{b}^t_{i}~~~~~~\text{       for each buyer }i \in [n]}_{\text{buyer-wise comparison}}
\end{equation}
\noindent In simpler terms, we establish that the revenue received by intermediate solution at round $t \geq 1$, is at least $50\%$ what RMVUP received at the same round. While we also establish that the overall spending of any buyer $i \in [n]$ in the Algorithm~\ref{alg:1} is at least $50\%$ of its overall spending in the intermediate solution. Combining the latter two claims the proof of Theorem~\ref{t:online} directly follows.

We dedicate the rest of the section to establish the \textit{time-wise} and \textit{buyer-wise} comparisons. Equations ~\ref{eq:17} is formally established in Lemma~\ref{l:5}. The proof of Lemma~\ref{l:5} crucially uses Theorem~\ref{t:revenue}.

\begin{lemma}\label{l:5}
  At each round \(t \geq 1\), the revenue of FPPE is at least $1/2$ the revenue of RMVUP,
  \[\sum_{i\in S(t)} \hat{b}^t_{i} \geq \frac{1}{2} \sum_{i\in S(t)}o^t_{i}.\]
\end{lemma}
\begin{proof}
  Let \((\mathbf{x^\star}^t,\mathbf{o^t})\) denote the allocation/payment vector of RMVUP at round \(t \in [T]\). Let \((\mathbf{\tilde{x}},\mathbf{\tilde{p}},\mathbf{\tilde{b}},\mathbf{\tilde{\alpha}})\) denote the first-price pacing equilibrium of a market composed by the buyers $i \in S(t)$ where each buyer $i \in S(t)$ has budget $\tilde{B}_i = o_i^t$. Since $\hat{B}^t_i \geq \tilde{B}_i$ by the \textit{budget monotonicity} property of first-price pacing equilibrium (Lemma~\ref{l:pacing} in Appendix~\ref{a:online}),
  \[\sum_{i\in S(t)} \hat{b}^t_{i} \geq \sum_{i\in S(t)}\tilde{b}_i.\]
  By the optimality condition of RMVUP, \((\mathbf{x^\star}^t,\mathbf{o^t})\) must be the revenue maximization solution of the market with $i \in S(t)$ and $\tilde{B}_i = o_i^t$. However by the definition of \((\mathbf{x^\star}^t,\mathbf{o^t})\) each buyer pays exactly $o_i^t$, Theorem~\ref{t:revenue} implies that
  \[\sum_{i \in S(t)} \tilde{b}_i \geq \frac{1}{2} \sum_{i \in S(t)} o_i^t.\]
\end{proof}
Up next we provide the \textit{buyer-wise} comparison guarantee (Equation~\ref{eq:19}) that is formally established in Lemma~\ref{l:main}. 
\begin{lemma}\label{l:main}
  The total spending of any buyer $i \in [n]$ in Algorithm~\ref{alg:1} is at least $1/2$ of the its total spending in the intermediate solution,
  \[
    \sum_{t=s_i}^{t_i} b_i^t \geq \frac{1}{2} \sum_{t=s_i}^{t_i} \hat{b}_i^t.
  \]
\end{lemma}
\begin{proof}
  First notice that the total spending of buyer $i \in [n]$  in the intermediate solution is at most $2 \cdot B_i$,
  \begin{align*}
    \sum_{t=s_i}^{t_i} \hat{b}_i^t \leq
    \sum_{t=s_i}^{t_i} \hat{B}_i^t = \sum_{t=s_i}^{t_i} \max(o_i^t, B_i^t)
    \leq \sum_{t=s_i}^{t_i} \left(o_i^t + B_i^t \right)
    \leq \sum_{t=s_i}^{t_i} o_i^t + \sum_{t=s_i}^{t_i}  B_i^t
    \leq 2 B_i.
  \end{align*}
  As a result, we only need to consider the buyers $i \in [n]$ such that
  \[ \sum_{t=s_i}^{t_i} b_i^t < B_i.\]

  Let us consider such a buyer $i \in [n]$. We will establish the even stronger claim that
  \begin{equation}\label{eq:123}
    b_i^t \geq \hat{b}^t_i \quad \text{ for each } t \in [s_i,t_i].
  \end{equation}
  The cornerstone idea is that since $\sum_{t=s_i}^{t_i} b_i^t < B_i$ then
  \[ b_i^t < B_i^t \quad \text{ for each } t \in [s_i,t_i].\]
  which then by Property~(6) of a first-price pacing equilibrium implies that
  \[ \alpha_i^t = 1 \quad \text{ for each } t \in [s_i,t_i].\]
  We will leverage the latter fact to establish Equation~\ref{eq:123}. To do so, we introduce two more \textit{intermediate First-Price Pacing Equilibria} with budget vectors \( B^{(1)} \) and \( B^{(2)} \).
  \smallskip

  Let \( (\mathbf{x}^t,\mathbf{b}^t,\mathbf{p}^t,\mathbf{\alpha^t}) \) denote the FPPE computed by Algorithm~\ref{alg:1} at time-step \( t \) with the budget vector \( B^t := (B_1^t,\ldots, B_n^t) \). Now consider the budget vector $B^{(1)}$ defined as follows,
  \[B_j^{(1)} = \begin{cases} B_j^t, & j \neq i \\ b_i^t, & j = i \end{cases} .\]
  Notice that in modified budget vector $B^{(1)}$, each buyer $j \in S(t)$ attains its original budget apart from buyer $i \in [n]$ whose budget is exactly its spending in original FPPE.
  \begin{claim}\label{c:1}
    \( (\mathbf{x}^t,\mathbf{b}^t,\mathbf{p}^t,\mathbf{\alpha}^t) \) remains the FPPE even under the modified budgets $B^{(1)}$. Moreover, $\alpha_i^t = 1$.
  \end{claim}
  \begin{proof}
    The only FPPE constraint of concern is the {\em budget feasibility } constraint for buyer \(i\) i.e. \(\sum_{j=1}^m x_{ij} \cdot p_j \leq B^{(1)}_i\), but by definition buyer \(i\) pays exactly \( \sum_{j=1}^m x_{ij} \cdot p_j = b_i^t = B^{(1)}_i\) which is now their budget.
  \end{proof}
  Then consider the effect of increasing the budgets of all other buyers \(j \neq i\) to the budget they have in the intermediate solution i.e. \( B^{(2)}_j = \begin{cases} \hat{B}_j^t, & i\neq j \\ b_i^t, & i=j \end{cases} \), and let \( (\mathbf{x}^{(2)},\mathbf{b}^{(2)},\mathbf{p}^{(2)},\mathbf{\alpha}^{(2)}) \) denote the FPPE for this instance. Notice that by the definition of $\hat{B}^t$ and $B^{(2)}$ the budgets are monotonically increasing \( B^{(2)}_j \geq B_j^{(1)} \). Thus, the pacing multipliers must also weakly increase \( \alpha _j^{(1)} \leq \alpha _j^{(2)} \) (see Lemma~\ref{l:pacing} in Appendix~\ref{a:online}), but recall that \( \alpha_i^{(1)} = \alpha _i^{t} = 1 \) and therefore \( \alpha _i^{(2)} = 1 \).
  We can now compare the FPPE for budgets \( B_j^{(2)} \) to the FPPE of the intermediate solution \( \hat{B} _j^t \).

  \begin{claim}\label{c:2} Since $\alpha_i^{(2)} = 1$ we have that
    \( (\mathbf{x}^{(2)},\mathbf{b}^{(2)},\mathbf{p}^{(2)},\mathbf{\alpha}^{(2)}) \) is also the FPPE for the intermediate solution, budget vector $\hat{B}^t = (\hat{B}_1,\ldots,\hat{B}_n)$.
  \end{claim}
  \begin{proof}
    Since the only difference between \( B^{(2)} \) and \( \hat{B}^t \) relate to buyer \( i \), we can focus on FPPE constraints relating to this buyer. Now, the only constraint that could be violated by an increase in budgets is the {\em no unnecessary pacing} constraint, however since \( \alpha _i^{(2)} = 1 \) this is not an issue. We must therefore conclude that \( ( \mathbf{x}^{(2)}, \mathbf{b}^{2}, \mathbf{p}^{(2)}, \mathbf{\alpha}^{(2)}) \) is also the FPPE for the intermediate solution.
  \end{proof}
  \noindent As result, we have that $\hat{b}_i^t = b^{(2)}_i \leq B_i^{(2)} = b_i^t$.
\end{proof}

\section{Revenue Guarantees for Concave Valuations}
\label{s:concave}

In this section, we extend the revenue guarantees of FPPE in the case of concave valuations, once buyer $i\in [n]$ gets quantity $x_{ij}> 0$ of good $j \in [m]$, its valuation is $v_{ij}(x_{ij})$. Concave valuations capture the fact that as the quantity $x_{ij} > 0$ increases the utility of buyer $i \in [n]$ exhibits diminishing returns. We assume that each valuation function $v_{ij}(\cdot)$ satisfies the following conditions:
\begin{enumerate}
  \item $v_{ij}(x_{ij})$ is non-decreasing.
  \item $v_{ij}(x_{ij})$ is concave\footnote{
        for any $x_{ij},x_{ij},\lambda \in [0,1]$,
        $v_{ij}( \lambda \cdot x_{ij} + (1 -\lambda) \cdot x'_{ij} ) \geq \lambda \cdot v_{ij}(x_{ij}) + (1 -\lambda)\cdot v'_{ij} (x'_{ij} )$} in $[0,1]$.
  \item $v_{ij}(0) = 0$.
\end{enumerate}

We first remark that computing the optimal revenue RMVUP in this case can be done by solving the convex program in which the individual rationality constraint $v_{ij} \cdot x_{ij} \ge b_i$ is replaced by the nonlinear yet convex constraint $v_{ij}(x_{ij}) \ge b_i$.

The notion of FPPE~\cite{CKPSSMW19} is taylored to the case of linear valuations, $v_{ij}(x_{ij}) = v_{ij} \cdot x_{ij}$ and it is not clear how it generalizes to general concave valuations. \cite{CKPSSMW19} establish that FPPE corresponds to the solution of the Eiseberg-Gale (EG) program for quasi-linear Fisher markets with linear valuations~\cite{CDGJMVY17}. As a result, a natural candidate for extending the notion of FPPE in general concave valuations is via extending the convex program of \cite{CDGJMVY17} to quasi-linear Fisher markets with concave valuations.

In order to extend the convex program of \cite{CDGJMVY17} in the case of concave valuation, we simply use the constraint $u_i \leq \sum_{j \in [m]} v_{ij}(x_{ij}) + \delta_i$ to derive the primal problem~(see Figure~\ref{f:EG}). In the primal program, the variable $\delta_i \ge 0$ denotes the \emph{budget excess}, i.e., the amount of money that buyer $i \in [n]$ does not spend.
The dual program can be derived using Fenchel duality; in the dual, $p_j$ represents the per-unit price of good $j$. We remark that as in the linear case, the solution of primal-dual pair of Figure~\ref{f:EG} will result in \textit{fixed unit-prices}. We also remark that the payment of an buyer $i \in [n]$, will be $\sum_{i \in [n]} \sum_{j \in [m]} x_{ij} \cdot p_j$.
\begin{figure}[h]

  \centering
  \setlength{\arraycolsep}{2pt}
  \scalebox{0.8}{
	\begin{minipage}{0.48\textwidth}
      \centerline{\textbf{Primal}}
      \[
        \begin{aligned}
          \max_{u,\delta,x}\quad
          & \sum_{i=1}^n \big(B_i \ln u_i - \delta_i\big)                       \\
          \text{s.t.}\quad
          & u_i \le \sum_{j=1}^m v_{ij}(x_{ij}) + \delta_i
          &                                                & \forall i \in [n], \\
          & \sum_{i=1}^n x_{ij} \le 1
          &                                                & \forall j \in [m], \\
          & u_i \ge 0,\ \delta_i \ge 0,\ x_{ij} \ge 0.
        \end{aligned}
      \]
	\end{minipage}
	\hfill
	\begin{minipage}{0.48\textwidth}
      \centerline{\textbf{Dual}}
      \[
        \begin{aligned}
          \min_{\alpha,p}\quad
          & \sum_{i=1}^n \Big(B_i \ln \frac{B_i}{\alpha_i} - B_i\Big) + \sum_{j=1}^m p_j                      \\
          & + \sum_{i=1}^n \sum_{j=1}^m
            \sup_{x \ge 0}
            \big(\alpha_i v_{ij}(x) - p_j x\big)                                                                 \\
          \text{s.t.}\quad
          & 0 < \alpha_i \le 1
          &                                                                              & \forall i \in [n], \\
          & p_j \ge 0
          &                                                                              & \forall j \in [m].
        \end{aligned}
      \]
	\end{minipage}}
  \caption{}
  \label{f:EG}
\end{figure}

In Theorem~\ref{t:approx} we establish that a solution produced by the EG program of Figure~\ref{f:EG} satisfies both the individual rationality and the budget spending constraints ($p_j \cdot x_{ij} \leq v_{ij}(x_{ij})$ and $\sum_{j \in [m]} p_j \cdot x_{ij} \leq B_i$). Most importantly, we establish that the revenue guarantees of a solution of the EG primal/dual programs.

\begin{theorem}\label{t:rev_concave}
  Let $(\mathbf{x},\mathbf{p})$ the solution to  primal/dual program of Figure~\ref{f:EG}. Both individual rationality and budget feasibility are guaranteed for each buyer $i \in [n]$ and the produced revenue satisfies,
  \[\sum_{ i \in [n]}\sum_{ j \in [m]} p_j \cdot x_{ij} \geq \frac{1}{\rho(\rho + 1)} \cdot \mathrm{RMVUP}\]
  where $\rho := \frac{v'_{ij}(0)}{v'_{ij}(1)}$ for general concave functions and $\rho:= \log\left(1 + \max_{ij} \frac{v'_{ij}(0)}{v'_{ij}(1)}\right)$ for concave functions satisfying the additional assumption that $x_{ij}\cdot v'(x_{ij})$ is a non-decreasing.
\end{theorem}
\noindent Theorem~\ref{t:rev_concave} generalizes the revenue guarantees of Theorem~\ref{t:revenue}. Notice that in case of linear function the ratio $\rho = 1$ and Theorem~\ref{t:rev_concave} recovers Theorem~\ref{t:revenue}. In case of general concave functions the ratio $\rho = \frac{v'_{ij}(0)}{v'_{ij}(1)}$ while in the special case where $x_{ij}\cdot v'(x_{ij})$ is a non-decreasing, the ratio $\rho:= \log\left(1 + \max_{ij} \frac{v'_{ij}(0)}{v'_{ij}(1)}\right)$ which is much smaller. The latter class of valuation functions captures the important class of sum of concave monomials, $v_{ij}(x) = \sum_{k} \lambda_k \cdot ( x - p_k)^{\alpha_k}$ where $\lambda_k,p_k \geq 0$ and $\alpha_k \in [0,1]$.

In the rest of the section, we provide the main steps for establishing Theorem~\ref{t:rev_concave}.
The main technical contribution of section is Theorem~\ref{t:approx} establishing crucial properties for any solution of $(\mathbf{x},\mathbf{p})$ of the primal/dual program of Fig.~\ref{f:EG} such as individual rationality, budget feasibility, and utility maximization.
\begin{restatable}{theorem}{thmApprox}\label{t:approx}
  Let $(\mathbf{x},\mathbf{p})$ denote the solutions of the primal/dual programs in Figure~\ref{f:EG}. Then the following properties hold,
  \begin{enumerate}
    \item $v_{ij}(x) \geq p_j \cdot x_{ij}$
          \smallskip

    \item $\sum_{j=1}^n p_j \cdot x_{ij} \leq B_i$
          \smallskip

    \item If $p_j > 0$ then $\sum_{i=1}^nx_{ij} = 1$
          \smallskip

    \item For each buyer $i \in [n]$ one of the following properties holds
          \smallskip
          \begin{enumerate}
            \item $\mathbf{x}_i \in \mathrm{argmax}_{x_i} \{\sum_{j=1}^m v_{ij}(x_{ij}) - p_j \cdot x_{ij}~:~\sum_{j \in [m]}p_{ij}\cdot x_{ij} \leq B_i\}$ (utility maximization)\\
                  \smallskip
                  or
                  \smallskip
            \item $\sum_{j=1}^n p_j \cdot x_{ij} \geq B_i / \rho$ where $\frac{v'_{ij}(0)}{v'_{ij}(1)}$ for general concave functions and $\rho:= \log\left(1 + \max_{ij} \frac{v'_{ij}(0)}{v'_{ij}(1)}\right)$ for concave functions satisfying $x_{ij}\cdot v'(x_{ij})$ be a non-decreasing function.
          \end{enumerate}
  \end{enumerate}
\end{restatable}
\noindent Due to space limitation the proof of Theorem~\ref{t:approx} is deferred to the Appendix~\ref{app:concave}.

Using Theorem~\ref{t:approx}, in Lemma~\ref{t:POA_LW_concave} we generalize Lemma~\ref{t:POA_LW}. In the case of concave valuation functions, the Liquid Social Welfare is defined as $\mathrm{LW}(\mathbf{x}):= \sum_{i \in [n]} \min \left(\sum_{j\in [m]} v_{ij}(x_{ij}) , B_i \right)$.

\begin{restatable}{lemma}{LemPOALWConcave}\label{t:POA_LW_concave}
  Let $(\mathbf{x},\mathbf{p})$ denote the solutions of the primal/dual programs in Figure~\ref{f:EG}.  and $\mathbf{o^\star}$ the allocation maximizing Liquid Social Welfare. Then, \[\mathrm{LW}(\mathbf{x}) \geq \frac{1}{\rho + 1}\cdot \mathrm{LW}(\mathbf{o^\star}).\]
\end{restatable}
\noindent The proof of Lemma~\ref{t:POA_LW_concave} is presented in Appendix~\ref{app:concave}. We conclude the section with the proof of Theorem~\ref{t:rev_concave} associating the revenue of $(\mathbf{x},\mathbf{p})$ with RMVUP. Let $(\mathbf{x^\star},\mathbf{b^\star})$ denote the RMVUP solution. 

We first approximately lower bound the 
spending of each buyer by its Liquid Welfare,
\begin{equation}\label{eq:101}
  \sum_{j \in [m]} p_{j}\cdot x_{ij} \geq \frac{1}{\rho} \cdot \min \left(\sum_{j=1}^m v_{ij}(x_{ij}) , B_i \right)\quad \quad \text{each buyer }i \in [n].
\end{equation}
Equation~(\ref{eq:101}) follows directly in case $\sum_{j \in [m]} p_{j}\cdot x_{ij} \geq B_i / \rho$, we thus assume $\sum_{j \in [m]} p_{j}\cdot x_{ij} < B_i / \rho$. Then Theorem~\ref{t:approx} implies that
\vspace{-2mm}
\[\mathbf{x}_i \in \mathrm{argmax}_{x_i} \left\{\sum_{j=1}^m v_{ij}(x_{ij}) - p_j \cdot x_{ij}~:~\sum_{j \in [m]}p_{ij}\cdot x_{ij} \leq B_i\right\}.\]
Since $\sum_{j \in [m]} p_j \cdot x_{ij} < B_i$ we get that $v'_{ij}(x_{ij}) = p_j$ for all goods $j \in [m]$. Due to the concavity of the $v_{ij}(\cdot)$ we have that
$ v_{ij}(x_{ij}) \leq v'_{ij}(0) \cdot x_{ij} \leq \frac{1}{\rho}\cdot v'_{ij}(x_{ij}) \cdot x_{ij} = \frac{1}{\rho}\cdot p_j  x_{ij}$ and Eq.~(\ref{eq:101}) follows by summing over all $j \in [m]$.
Then the proof of Theorem~\ref{t:rev_concave} follows by the fact that
\[\sum_{i \in [n]}\sum_{j \in [m]} p_j \cdot x_{ij} \geq \frac{1}{\rho}\sum_{i = 1}^n \min \left(\sum_{j=1}^m v_{ij}(x_{ij}) , B_i \right) = \frac{1}{\rho} \cdot \mathrm{LW}(\mathbf{x}) \geq \frac{1}{\rho(\rho + 1)} \cdot \mathrm{LW}(\mathbf{x^\star}) \geq \mathrm{RMVUP}.\]
\section{Conclusion}
In this work, we analyze the revenue properties of the First-Price Pacing Equilibrium (FPPE) in markets with divisible goods and budget-constrained buyers. Our results show that FPPE guarantees at least $50\%$ of the optimal revenue, even when compared against benchmarks that allow variable unit prices, while itself relying only on fixed unit prices. This guarantee is especially meaningful in light of our inapproximability result: revenue maximization under fixed unit prices is APX-hard. As a consequence, FPPE is the currently best-known polynomial-time approximation algorithm for revenue maximization with fixed unit prices. We further study FPPE in dynamic environments with online buyer arrivals. We establish that recomputing an FPPE at each round based on the set of active buyers yields a $1/4$-competitive online algorithm.

Several open questions remain. On the offline side, there is a gap between the APX-hardness of revenue maximization under fixed unit prices and the $1/2$-approximation factor achieved by FPPE; determining whether this approximation ratio is optimal or whether stronger inapproximability results can be obtained is an important open problem. On the online side, there is a wide gap between the current $1/4$-competitive guarantee and the upper bound of approximately $0.85$. We conjecture that FPPE achieves a $1/2$-competitive ratio in the online setting. Resolving these gaps would further establish FPPE as a canonical mechanism for revenue-robust autobidding design.

\bibliographystyle{plainnat}
\bibliography{competitive}

\clearpage
\appendix

\section{Rounding Algorithm and Proof of Lemma~\ref{l:1}}\label{app:red}
In this section we present the rounding algorithm rounding a solution of RMFUP to another solution with fixed prices satisfying the conditions stated in Lemma~\ref{l:1}.

\begin{algorithm}[h!]
  \caption{Rounding the FPRM solution}
  \label{alg:round}

  \KwIn{$(\mathbf{b},\mathbf{p})$}

  \vspace{1mm}
  \textbf{\textcolor{blue}{Phase I: Increase prices}}\\
  $(\mathbf{\hat{b}},\mathbf{\hat{p}}) \leftarrow (\mathbf{b},\mathbf{p})$\;

  \ForAll{goods $s \in [m]$}{
    \smallskip
    \smallskip

    $
    \hat{p}_s :=
    \begin{cases}
      2/3, & p_s \leq 2/3,      \\
      1,   & p_s \in (2/3 , 1], \\
      1,   & p_s >1.
    \end{cases}
    $
  }

  \vspace{1mm}
  \textbf{\textcolor{blue}{Phase II: Handle goods with special buyers}}\\
  $(\mathbf{\hat{b}},\mathbf{\hat{p}}) \leftarrow (\mathbf{b},\mathbf{p})$\;

  \ForAll{goods $s \in [m]$}{
    \If{$\sum_{e \in E} b_{es} \leq 2/3$}{
      Assign price $\hat{p}_s = 2/3$ and sell exclusively to the special buyer
      $s$, setting $\hat{b}_{ss} = 2/3$\;
    }
  }

  \vspace{1mm}
  \textbf{\textcolor{blue}{Phase III: Round element buyers' allocations}}\\
  $(\mathbf{\hat{b}},\mathbf{\hat{p}}) \leftarrow (\mathbf{b},\mathbf{p})$\;

  \ForAll{element buyers $e \in E$}{
    \If{$\sum_{s \in [m]} b_{es} > 0$}{
      Let $s \in [m]$ be an arbitrary good such that $b_{es} > 0$\;
      Set $\hat{b}_{es} = 1/3$ and $\hat{b}_{es'} = 0$ for all $s' \neq s$\;
    }
  }

  \vspace{1mm}
  \textbf{\textcolor{blue}{Phase IV: Adjust goods with partial allocations}}\\
  $(\mathbf{\hat{b}},\mathbf{\hat{p}}) \leftarrow (\mathbf{b},\mathbf{p})$\;

  \ForAll{goods $s \in [m]$}{
    \If{$p_s = 1$ \textbf{and} $\sum_{e \in E} b_{es} \leq 2/3$}{
      Assign price $\hat{p}_s = 2/3$ and sell exclusively to the special buyer
      $s$, setting $\hat{b}_{ss} = 2/3$\;
    }
  }
  \KwOut{$(\mathbf{\hat{b}},\mathbf{\hat{p}})$}
\end{algorithm}
Algorithm~\ref{alg:round} proceeds in $4$ different phases. In Phase~I, the algorithm sets all prices that are less than $2/3$ either to $2/3$. The latter does not affect the individual rationality since all valuations are greater than $2/3$. The algorithm also sets all prices that are greater than $2/3$ to $1$. The latter also does not affect individual rationality since these goods can only be sold to the non-special buyers that have valuation either $0$ or $1$. Finally setting a price strictly greater than $1$ means that the price is not sold at all. So decreasing price does neither affect any constraint.

In Phase~$\mathrm{II}$ if the algorithm finds that the overall revenue produced by a given good $\sum_{i \in [n]} b_{is} \leq 2/3$ then the algorithm sells the whole good to the special buyer $s \in [n]$ at price $2/3$. The latter guarantees that at the end of Phase~$\mathrm{I}$ $i)$ any good $s \in [m]$ is either sold to the respective special buyer $s \in [n]$ or is sold to only to element buyers $ii)$ the overall revenue of the produced solution $(\mathbf{\hat{b}},\mathbf{\hat{p}})$ does not decrease.

In Phase~$\mathrm{III}$, Algorithm~\ref{alg:round} iterates over the element buyers and in case it finds an element buyer $e \in E$ that spends a positive fraction of its budget $(\sum_{e \in E}b_{es} > 0)$ then Algorithm~\ref{alg:round} selects an arbitrary element $e \in E$ such that $b_{es} > 0$. Due to Phase~$\mathrm{II}$ we are ensured that no good is simultaneously sold to both a special and the an element buyer. Thus, we are ensured that $p_s = 1$. As a result, by spending its whole budget $B_e = 1/3$ buyer $e \in E$ only gets $1/3$ of the good $s \in [m]$. And since its good $s \in [m]$ can be sold to exactly three element buyers, we are ensured that at the end of Phase~$\mathrm{II}$ the solution  $(\mathbf{\hat{b}},\mathbf{\hat{p}})$ is feasible (no overselling). Also the revenue from Phase~$\mathrm{I}$ to Phase~$\mathrm{II}$ does not decrease since all buyer spend as at least what they were spending in Phase~$\mathrm{II}$. As a result, at the end of Phase~$\mathrm{II}$ we are ensured that
\begin{itemize}
  \item Each element buyer $e \in E$ either spends its whole budget $B_e = 1/3$ to a single good $s \in [m]$ where $e \in s$ or does not spend any of its budget.

  \item Any good $s \in [m]$ either has price $p_s = 2/3$ and its completely sold to the special buyer $s \in [n]$ or $p_s = 1$ and the good is sold to $1$ or $2$ or $3$ buyers with each one of them contributing either its whole budget ($1/3$) or $0$.

\end{itemize}
Phase~$\mathrm{IV}$ guarantees that if a good $s \in [m]$ admits price $1$ then it is sold to exactly $3$ element buyers. The latter is very easy to achieve. In case after Phase~$\mathrm{III}$ there exists a good $s \in [m]$ with $p_s = 1$ but $\sum_{e \in E}b_{es} \leq 2/3$. Then Algorithm~\ref{alg:round} simply reduces the price of $s \in [m]$ to $2/3$ and completely sells it to the special buyer. The latter does not decrease the revenue at the same guarantees that if $p_s = 1$ then exactly three element buyers pay $1/3$ for good $s \in [m]$.

\section{Omitted Proof of Section~\ref{s:revenue_linear}}\label{app:2}

\RMFUPBound*

\begin{proof}
  Consider the case of a single good, $m=1$ and $n$ buyers. Each buyer $i \in [n]$ admits budget,
  \[B_i = \begin{cases} n & i = 1 \\ 1 & i > 1 \end{cases} \]
  and valuation \[v_{i1} = n \cdot  B_i.\]

  Notice that in RMVUP the seller can sell $1/n$ of the good to buyer $1$ and charge her $b_{11} = n$. In other words, the unit-price of buyer $1$ is $n^2$. Notice that the latter does not violate neither the individual rationality nor the budget constraints. The seller can then sell $1/n$ of the good to each buyer $i > 1$, $x_{i1} = 1/n$ and charge her $b_{i1} = 1$. As a result,
  \[\mathrm{RMVUP} = 2n - 1.\]
  Let us now consider the case of fixed unit price across buyers. Let $p$ denote the unit price of the good. In case $ p \leq n $ then $\mathrm{RMFUP} \leq n$. In case $p > n$ then due to the individual rationality constraints only buyer $1$ is willing to buy a positive quantity of the good at this price. In this case $\mathrm{RMFUP} \leq n$ since buyer $1$ can spend at most his budget $B_1 = n$. So in this construction, the ratio of RMFUP to RMVUP is \(\frac{n}{n+(n-1)}\) and
  \[\lim_{n \to \infty} \frac{n}{n+(n-1)} = \frac{1}{2}.\]
\end{proof}

\section{Omitted proof of Section~\ref{s:online}}\label{a:online}
The proof of Theorem~\ref{t:any-competitive} the competitive ratio of any algorithm is no better than \( \frac{1}{4} \left( 2 + \sqrt{2} \right) \approx 0.85 \).
\begin{proof}\label{p:any-competitive}
  \begin{figure}[h]
    \begin{tikzpicture}[
      buyer/.style={circle, draw, minimum size=1.5cm, fill=blue!10, thick},
      item/.style={rectangle, draw, minimum size=1.5cm, fill=red!10, thick},
      valuation/.style={font=\normalsize,fill=white, inner sep=4pt},
      ]
      \def\BudgetX{-1.5}
      \def\BuyerX{0.5}
      \def\ItemX{6}
      \def\TimeX{7.5}
      \def\ValuationX{3}
      \def\LineMinX{-2.5}
      \def\LineMaxX{8.5}
      \def\NodeFontSize{\normalsize}

      \def\buyerOneBudget{1}
      \def\buyerTwoBudget{$1+\sqrt{2}$}

      \node at (\BudgetX, 3.5) {\textbf{Budget $B_i$}};
      \node at (\BuyerX, 3.5) {\textbf{Buyer $i$}};
      \node at (\ItemX, 3.5) {\textbf{Good $j$}};
      \node at (\ValuationX, 3.5) {\textbf{Valuation $v_{ij}$}};
      \node at (\TimeX, 3.5) {\textbf{Time $t$}};

      \node[buyer, font=\NodeFontSize] (A1) at (\BuyerX, 2.0) {1};
      \node[anchor=center, font=\NodeFontSize] at (\BudgetX, 2.0) {\buyerOneBudget};

      \node[buyer, font=\NodeFontSize] (A2) at (\BuyerX, 0.0) {2};
      \node[anchor=center, font=\NodeFontSize] at (\BudgetX, 0.0) {\buyerTwoBudget};

      \draw[dashed, gray, thick] (\LineMinX, -1.5) -- (\LineMaxX, -1.5);
      \node at (\BudgetX, -1.5) [anchor=north west, font=\NodeFontSize] {Potential Arrival};

      \node[buyer, font=\NodeFontSize] (A3) at (\BuyerX, -3.0) {3};
      \node[anchor=center, font=\NodeFontSize] at (\BudgetX, -3.0) {\buyerTwoBudget};

      \node[item, font=\NodeFontSize] (I1) at (\ItemX, 2.0) {1};
      \node[anchor=center, font=\NodeFontSize] at (\TimeX, 2.0) {1};

      \node[item, font=\NodeFontSize] (I2) at (\ItemX, 0.0) {1};
      \node[anchor=center, font=\NodeFontSize] at (\TimeX, 0.0) {2};

      \draw[thick] (A1) -- node[valuation, above] {\buyerOneBudget} (I1);

      \draw[thick] (A2) -- node[valuation, above, sloped] {\buyerTwoBudget} (I1);

      \draw[thick] (A2) -- node[valuation, below] {\buyerTwoBudget} (I2);

      \draw[thick] (A3) -- node[valuation, above, sloped] {\buyerTwoBudget} (I2);

    \end{tikzpicture}
	\caption{\label{fig:difficult-instance}Adversarial instance which is difficult to solve well in an online manner.}
  \end{figure}
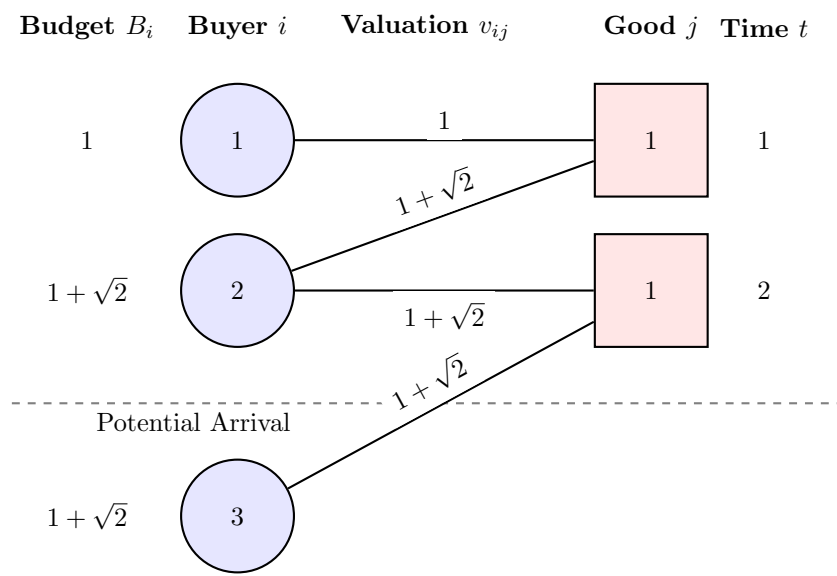

  Consider the problem instance illustrated in Figure~\ref{fig:difficult-instance}, in which buyers \( 1 \) and \( 2 \) have arrived at time-step \( t=1 \) and buyer \( 3 \) might arrive in time-step \( t=2 \). Due to the nature of online problems, any algorithm has to make an unalterable decision at time-step \( t=1 \), however, the best course of action at time \( t=1 \) depends on whether buyer \( 3 \) arrives.

  In this adversarial instance buyer \( 3 \) arrives only if the algorithm made a decision that would not be optimal if buyer \( 3 \) arrived. We therefore ensure that the decision made by the algorithm is not optimal, and by doing so we get a lower-bound for the competitive ratio of any online algorithm of \( \frac{1}{4} \left( 2 + \sqrt{2} \right) \).

  If the algorithm decides to allocate less than half the good at time-step 1 to buyer 2 i.e. \( x^1_{21} < 1 \), we allow buyer 3 to arrive, otherwise buyer 3 does not arrive.

  \begin{claim}
    If \( x^1_{21} \geq \frac{1}{2} \) then the revenue of any algorithm is at most \( \frac{3}{2} + \sqrt{2}\), but the RMVUP is \(2+\sqrt{2}\) and so the ratio is \( \frac{\frac{3}{2} + \sqrt{2}}{ 2+\sqrt{2} } = \frac{1}{4} \left( 2 + \sqrt{2} \right) \).
  \end{claim}
  \begin{proof}
    The best revenue when more than half the good at the first time-step is allocated to buyer \( 2 \) i.e. \( x_{21}^1 \geq 1 \) is to sell it at the highest price buyer \( 2 \) is willing to pay i.e. \( p_{21}^1 = \left( 1 + \sqrt{2} \right) \). The rest of the good can then be sold to buyer \( 1 \) at price \( 1 \) so revenue \( x_{11} \cdot p^1_{11} = \frac{1}{2} \cdot 1 = \frac{1}{2} \). The rest of buyer \( 2 \)'s budget can then be exhausted in time-step \( 2 \). The total revenue is thus \( \frac{1}{2}+ \left( 1 + \sqrt{2} \right) \).

    The RMVUP is to sell the entire good at time-step 1 to buyer 1 for the price of 1 and the entire good at time-step 2 to buyer 2 for a price of \( 1+\sqrt{2} \). The total revenue is thus \( 2+\sqrt{2} \).

    The ratio in this case is therefore \( \frac{\frac{3}{2} + \sqrt{2}}{2 + \sqrt{2}}\).
  \end{proof}

  However if the algorithm made the opposite choice the ratio is still the same.

  \begin{claim}
    If \( x^1_{21} < \frac{1}{2} \) then the revenue of any algorithm is at most \( \frac{1}{2} + \frac{3}{2} \sqrt{2}\), but the RMVUP is \( 2 \cdot (1+\sqrt{2}) \) and so the ratio is \( \frac{ \frac{1}{2} + \frac{3}{2} \sqrt{2}}{ 2 \cdot (1+\sqrt{2}) } = \frac{1}{4} \left( 2 + \sqrt{2} \right) \).
  \end{claim}
  \begin{proof}
    The best revenue that can be reached when less than half the good is sold to buyer 2 at time-step \( 1 \) is \( 2 \) i.e. \( x^1_{21} < \frac{1}{2} \) is like before to sell just less than half the good to buyer \( 2 \) at price \( p^1_{21} = 1+\sqrt{2} \), then sell the rest to buyer \( 1 \) at price \( p_{11}^1 = 1 \) for revenue \( x_{11}^1 p_{11}^1 = \frac{1}{2} \). The revenue at time-step \( 1 \) is thus \( \frac{1}{2} \left( 1 + 1 + \sqrt{2} \right) \).

    With the arrival of buyer \( 3 \) we can sell the good at time-step 2 for a price of \( 1+\sqrt{2} \) to both buyer 2 and buyer 3 at a revenue of \(\left(x_{21}^2 + x_{31}^2\right) \left( 1+\sqrt{2} \right) =  1+\sqrt{2} \), but no higher due to {\em individual rationality}.

    The total revenue is therefore \( \frac{1}{2} + \frac{1}{2} \left( 1 + \sqrt{2} \right) +  \left( 1 + \sqrt{2} \right)\), however the RMVUP is to sell the entire good at time-slot 1 to buyer 2 at price \( 1+\sqrt{2} \)  and the entire good at time-slot 2 to buyer 3 also at price \( 1+\sqrt{2} \) giving an optimal revenue of \( 2\cdot \left( 1+\sqrt{2} \right) \).

    The competitive ratio is thus \( \frac{\frac{1}{2} + \frac{3}{2} \left( 1+\sqrt{2} \right)}{2\cdot \left( 1+\sqrt{2} \right)
    } = \frac{1}{4} \left( 2 + \sqrt{2} \right) \).
  \end{proof}
  So no matter what choice the algorithm makes at time-step \( t=1 \), the ratio is \( \frac{1}{4} \left( 2 + \sqrt{2} \right) \approx 0.85\).
\end{proof}
This proof is very similar to the proof of weakly increasing revenue based on increasing the budget of an buyer~\cite{CKPSSMW19}, we simply focus on the pacing multipliers instead of the revenue.
\begin{lemma}\label{l:pacing}
  In an FPPE, increasing the budget of buyers weakly increases their pacing multipliers.
\end{lemma}
\begin{proof}
  Let \((\alpha,p,x)\) be an FPPE for the original budgets.
  Notice the only constraint which \((\alpha,p,x)\) might not satisfy when budgets get increased is {\em no unnecessary pacing}.
  So to find the FPPE we might need to weakly increase the pacing multipliers and thus the prices.
\end{proof}
\section{Omitted Proof of Section~\ref{s:concave}}\label{app:concave}

\thmApprox*
\begin{proof}
  We use the KKT conditions of the EG program and its Fenchel dual.
  Let $(x,u,\delta)$ and $(\alpha,p)$ be optimal primal and dual solutions.
  By the KKT conditions, we get the following set of equations:
  \begin{align}
    \frac{B_i}{u_i}                                      & = \alpha_i, \qquad 0<\alpha_i\le 1, \label{eq:kkt1} \\
    \alpha_i v'_{ij}(x_{ij})                             & \le p_j,
                                                      \quad \text{with equality if } x_{ij}>0, \label{eq:kkt2}                                                   \\
    \alpha_i\Big(u_i-\sum_j v_{ij}(x_{ij})-\delta_i\Big) & =0, \label{eq:kkt3}                                 \\
    p_j\Big(1-\sum_i x_{ij}\Big)                         & =0, \label{eq:kkt4}                                 \\
    \delta_i(1-\alpha_i)                                 & =0. \label{eq:kkt5}
  \end{align}

  \medskip
  \noindent
  \begin{enumerate}
    \item Let $x_{ij} > 0$ then by Equation~\eqref{eq:kkt2} we are ensured that
          \[\alpha_i \cdot v'_{ij}(x_{ij}) = p_j\]
          Since $\alpha_i \leq 1$ we are ensured that $v'_{ij}(x_{ij}) \geq p_j$ and by the concavity of $v_{ij}(\cdot)$ and the fact that $v_{ij}(0) = 0$, we get that
          \[v_{ij}(x_{ij}) \geq x_{ij} \cdot v'_{ij}(x_{ij}) \geq x_{ij} \cdot p_j\]
          Using \eqref{eq:kkt2} and $\alpha_i\le 1$ we finally get
          $v_{ij}(x_{ij})
          \geq
          p_j \cdot x_{ij}.$
          \smallskip
          \smallskip

    \item In case $\alpha_i=1$ then Equation~\eqref{eq:kkt1} implies that $u_i=B_i$.  Then by (1) we have that
          \[
          \sum_{j \in [m]} p_j x_{ij}
          \le
          \sum_{j \in [m]} v_{ij}(x_{ij})
          \le
          u_i
          =
          B_i.
          \]
          If $\alpha_i<1$ then Equation~\eqref{eq:kkt5} implies that
          $\delta_i = 0$. Then by Equation~\eqref{eq:kkt3} we have that
          \[u_i = \sum_{j \in [m]} v_{ij}(x_{ij}).\] Then by \eqref{eq:kkt1} we have that $ \sum_{j} v_{ij}(x_{ij}) > B_i$. By Equation~\eqref{eq:kkt2} we get that,
          \[
          \sum_{j \in [m]}p_j\cdot x_{ij} = \alpha_i\cdot \sum_{j \in [m]}  v'_{ij}(x_{ij}) \cdot x_{ij} \leq
          \alpha_i \cdot \sum_{j \in [m]} v_{ij}( x_{ij} ) = \alpha_i \cdot u_i = B_i.
          \]

          \smallskip
          \smallskip
          \smallskip
          \smallskip

    \item By Equation~\eqref{eq:kkt4} we have that if $p_j>0$ then $1-\sum_{i\in [n]} x_{ij}=0$.

          \smallskip
          \smallskip
          \smallskip
          \smallskip

    \item If $\alpha_i=1$ then $x_{ij}:= \mathrm{argmax}_{x \geq 0}(v_{ij}(x) - p_j x)$. If $\alpha_i<1$
          then $\delta_i=0$ and by Equation~\eqref{eq:kkt1} we have that $u_i=B_i/\alpha_i$. By concavity of $v_{ij}(\cdot)$ we have that
          \[
          v'_{ij}(1) \cdot x_{ij}
          \le
          v_{ij}(x_{ij})
          \le
          v'_{ij}(0)\cdot x_{ij}.
          \]
          By Equation~\eqref{eq:kkt2} we get that for $x_{ij} > 0$,
          \[
          p_j = \alpha_i\cdot  v'_{ij}(x_{ij}) \geq
          \alpha_i \cdot v'_{ij}(1) \geq \frac{1}{\rho} \alpha_i \cdot v'_{ij}(0).
          \]
          Hence,
          \[
          \sum_{j \in [m]} p_j x_{ij}
          \ge
          \frac{1}{\rho} \cdot \alpha_i \sum_{j \in [m]} v_{ij}(x_{ij})
          =
          \frac{1}{\rho} \cdot \alpha_i u_i
          =
          \frac{B_i}{\rho}.
          \]
          We now cosnider the case where the function $x_{ij}\cdot v'_{ij}(x_{ij})$ is non-decreasing. In this case we have that
          \begin{eqnarray*}
            v_{ij}(x_{ij}) &=& \int_{0}^{x_{ij}}v'_{ij}(x)~\partial x\\
                           &=& \int_{0}^z v'_{ij}(x)~\partial x + \int_{z}^{x_{ij}}v'_{ij}(x)~\partial x\\
                           &\leq& z \cdot v'_{ij}(0) + \int_{z}^{x_{ij}}v'_{ij}(x)~\partial x ~~~\quad \text{by concavity}\\
                           &\leq& z \cdot v'_{ij}(0) + x_{ij}\cdot v'_{ij}(x_{ij}) \int_{z}^{x_{ij}}\frac{1}{x}~\partial x ~~~\quad \text{ }x_{ij}v'(x_{ij}) \text{ is not decreasing}\\
                           &=& z \cdot v'_{ij}(0) + x_{ij}\cdot v'_{ij}(x_{ij})\cdot \log\left(\frac{x_{ij}}{z} \right) \text{ for any }z \in [0,x_{ij}]
          \end{eqnarray*}
          As a result, we have that
          \[ v_{ij}(x_{ij}) \leq \min_{z \in [0,x_{ij}]}\left\{z \cdot v'_{ij}(0) + x_{ij}\cdot v'_{ij}(x_{ij})\cdot \log\left(\frac{x_{ij}}{z} \right) \right \}\]
          Notice that the quantity $z \cdot v'_{ij}(0) + x_{ij}\cdot v'_{ij}(x_{ij})\cdot \log\left(\frac{x_{ij}}{z} \right)$ is minimized for $z := \frac{v'_{ij}(0)}{v'_{ij}(0)} \cdot x_{ij}$. As a result, we get that
          \[ v_{ij}(x_{ij}) \leq   x_{ij}\cdot v'_{ij}(x_{ij})\cdot\left( 1+  \log\left(\frac{v'_{ij}(0)}{v'_{ij}(1)} \right) \right) \]

          \noindent As before the last equation implies that
          \[\sum_{j \in [m]} x_{ij} \cdot p_j =\sum_{j \in [m]}  x_{ij}\cdot v'_{ij}(x_{ij}) \geq \sum_{j \in [m]} \frac{v_{ij}(x_{ij})}{\rho} \geq \frac{1}{\rho} \alpha_i \cdot \sum_{j \in [m]} v_{ij}(x_{ij}) = \frac{1}{\rho} \cdot \alpha_i \cdot u_i = \frac{B_i}{\rho}\]
  \end{enumerate}
\end{proof}
\LemPOALWConcave*
\begin{proof}
  To simplify notation let $v_i(x_i) := \sum_{j\in [n]} v_{ij} (x_{ij})$. Without loss of generality, we assume that $v_i(o^\star_i)\leq B_i$ for every buyer $i\in [n]$ since otherwise we can simply decrease an $x_{ij} > 0$.\\
  \smallskip
  \smallskip

  \noindent Let $V\subseteq [n]$ denote the subset of buyers
  that spend at most $B_i / \rho$. In particular,
  \[ V:= \{i \in [n]~:~ \sum_{j \in [m]} p_j \cdot x_{ij} \leq B_i /\rho\}.\]
  Any buyer $i \notin V$ spends at least $B_i /\rho$ of his budget, meaning that by individual rationality, for each buyer $i \notin V$,
  \[\sum_{j \in [m]} v_{ij}(x_{ij}) \geq \frac{B_i}{\rho}.\]
  Then by definition of the Liquid Social Welfare,
  \begin{align}\label{eq:LW-x}
    \lw(\x) & = \sum_{i\in [n]}{\min\left\{B_i,v_i(x_i)\right\}} \geq \frac{1}{\rho} \cdot \sum_{i\in V}{B_i}+\sum_{i\in [n]\setminus V}{v_i(x_i)}
  \end{align}
  and
  \begin{align}\label{eq:LW-xstar}
    \lw(\mathrm{o}^\star) & =\sum_{i\in [n]} \min\{B_i,{v_i(o^\star_i)}\} \leq \sum_{i\in V}{B_i}+\sum_{i\in [n]\setminus V}{v_i(o^\star_i)}
  \end{align}
  \noindent Given an buyer $i\in V$. We distinguish between two cases. If $v_i(o^\star_i)< \sum_{j\in [m]}{o^\star_{ij}\cdot p_j}$ then
  \begin{align}\label{eq:utility-positive}
    v_i(x_i)-\sum_{j\in [m]}{x_{ij}\cdot p_j} & \geq 0 >v_i(o^\star_i)-\sum_{j\in [m]}{o^\star_{ij}\cdot p_j}.
  \end{align}
  If $v_i(o^\star_i)\geq \sum_{j\in [m]}{o^\star_{ij}\cdot p_j}$ then our assumption about the allocation $\mathrm{o}^\star$ yields
  \begin{equation}\label{eq:100}
    \sum_{j\in [m]}{o^\star_{ij}\cdot p_j} \leq v_i(o^*_i) \leq B_i
  \end{equation}
  \noindent The latter means that the allocation $\mathbf{o}^\star$ satisfies the budget constraint for buyer $i \in [n]$. Since buyer $i \in V$ spends at most $B_i/\rho$, then Property~(4b) of Theorem~\ref{t:approx} implies that

  \begin{align}\label{eq:utility-maximizing}
    v_i(x_i)-\sum_{j\in [m]}{x_{ij}\cdot p_j} & \geq v_i(o^\star_i)-\sum_{j\in [m]}{o^\star_{ij}\cdot p_j}.
  \end{align}
  As a result, for any buyer $i \in V$ we get that
  \begin{align}\label{eq:competitive-equilibrium-condition}
    v_i(o^\star_i) & \leq v_i(x_i)+\sum_{j\in [m]}{p_j\cdot (o^\star_{ij}-x_{ij})},
  \end{align}
  Using Equation~(\ref{eq:competitive-equilibrium-condition}), Equation~(\ref{eq:LW-xstar}) yields
  \begin{align}\nonumber
    \lw(\mathrm{o}^\star) & \leq \sum_{i\notin V}{B_i}+\sum_{i\in V}{v_i(x_i)}+\sum_{i\in  V}{\sum_{j\in [m]}{p_j\cdot (o^\star_{ij}-x_{ij})}} \\\label{eq:bound-liquid-plus-prices}
             & \leq \rho \cdot \lw(\x)+\sum_{i\in  V}{\sum_{j\in [m]}{p_j\cdot (o^*_{ij}-x_{ij})}}.
  \end{align}
  The last equality is due to Equation~(\ref{eq:LW-x}). Finally, we show that
  \begin{align}\label{eq:bound-for-prices}
    \sum_{i\in [n]\setminus V}{\sum_{j\in [m]}{p_j\cdot (o^\star_{ij}-x_{ij})}} & \leq \lw(\x).
  \end{align}
  The proof of the theorem will then follow by Equation~(\ref{eq:bound-liquid-plus-prices}). We have
  \begin{align*}
    \sum_{i\in  V}{\sum_{j\in [m]}{p_j\cdot (o^\star_{ij}-x_{ij})}} & = \sum_{j\in [m]}{p_j \sum_{i\in  V}{(o^\star_{ij}-x_{ij})}}         \\
                                                  & \leq \sum_{j\in [m]}{p_j\cdot \left(1- \sum_{i\in V}{x_{ij}}\right)} \\
                                                  & = \sum_{j\in [m]}{p_j\sum_{i\notin V}{x_{ij}}}
                                                    = \sum_{i\notin V}{\sum_{j\in [m]}{p_j\cdot x_{ij}}}                                                                                   \\
                                                  & \leq \sum_{i\notin V}{B_i}
                                                    \leq \lw(\x),
  \end{align*}
  as desired by Equation~(\ref{eq:bound-for-prices}). The first and third equalities are obvious. The first inequality follows by the feasibility of allocation $\mathbf{o}^\star$ which implies that $\sum_{i\in V}{o^\star_{ij}}\leq 1$. The second inequality is due to the budget constraint for buyer $i\in V$ (i.e., $\sum_{j\in [m]}{p_j\cdot x_{ij}}\leq B_i$) and the third one is due to Equation~(\ref{eq:LW-x}). The second equality is due to the fact that $p_j\cdot \left(1-\sum_{i\notin V}{x_{ij}}\right)=p_j\sum_{i\in V}{x_{ij}}$ for every good $j\in [m]$. To see why this holds, notice that this is trivially the case if $p_j=0$. Otherwise, since $\x$ is a competitive equilibrium, the condition $p_j>0$ for good $j\in [m]$ implies that $\sum_{i\in [n]}{x_{ij}}=1$ and, hence, $1-\sum_{i\in V}{x_{ij}} = \sum_{i\notin V}{x_{ij}}$.
\end{proof}

\end{document}